\documentclass[11pt]{amsart}

\usepackage[normalem]{ulem}

\usepackage{latexsym} 
\usepackage{amsmath} 
\usepackage{epsfig}
\usepackage{amssymb}
\usepackage{enumerate}
\usepackage{color}
\usepackage{graphicx}
\usepackage{hyperref}
\usepackage[foot]{amsaddr}

\newtheorem{theorem}{Theorem}[section]

\newtheorem{lemma}[theorem]{Lemma}
\newtheorem{sublemma}{}[theorem]
\newtheorem{corollary}[theorem]{Corollary}

\newtheorem{observation}[theorem]{Observation}
\newtheoremstyle{problem}{}{}{}{0pt}{}{}{0pt}{}
\theoremstyle{problem}
\newtheorem*{prb}{}

\newcommand{\cD}{{\mathcal D}}
\newcommand{\cG}{{\mathcal G}}
\newcommand{\cI}{{\mathcal I}}
\newcommand{\cL}{{\mathcal L}}

\newcommand{\NTstO}{{\sc NT-st-Orientation}}
\newcommand{\NAESAT}{{\sc Not-All-Equal 3-SAT}}
\newcommand{\RNAESAT}{{\sc Positive NAE-3-SAT-E4}}
\newcommand{\PhyloNTO}{{\sc Phylogenetic NT-Orientation}}

\title[]{Orientations without transitive arcs for cubic graphs and phylogenetic networks}

\author{Janosch Döcker$^1$ and Simone Linz$^2$} 

\thanks{The authors were supported by the New Zealand Marsden Fund from
Government funding, administered by the Royal Society Te Ap\=arangi New Zealand.}

\address{$^{1,2}$ School of Computer Science, University of Auckland, Auckland, New Zealand}
\email{s.linz@auckland.ac.nz}


\date{\today}

\begin{document}
 
\begin{abstract}
An  $st$-orientation of an undirected graph $G$ is an acyclic digraph with a single source $s$ and a single sink $t$ that can be obtained from $G$ by assigning a direction to each edge. The classical problem of deciding if an undirected graph $G$ has an $st$-orientation can be solved efficiently. On the other hand, deciding if an $st$-orientation of $G$ exists that does not have any transitive arc is NP-complete, even if each vertex of $G$ has degree at most four.
Here we show that this last decision problem remains NP-complete if $G$ is cubic, which settles an open question by Binucci et al. (2025). We obtain NP-completeness for two variants of the problem: (i) $s$ and $t$ are fixed and given as part of the input and (ii) $s$ and $t$ can be chosen freely.  We then use these results to investigate the computational complexity of a problem that arises in computational evolution. Specifically, we show that the problem of deciding if an unrooted binary phylogenetic network has an orientation as a rooted binary phylogenetic network without any shortcuts (the analog of a transitive arcs in phylogenetics) is NP-complete. Our results connect the two (mostly) distinct research areas of orienting undirected graphs and orienting unrooted phylogenetic networks.
\end{abstract}

\maketitle

\section{Introduction}

An \emph{$st$-orientation} (also called a \emph{bipolar orientation}) of an undirected graph $G$ is an acyclic digraph with a single source $s$ and a single sink $t$ that can be obtained from $G$ by assigning a direction to each edge. Deciding if an undirected graph has an $st$-orientation is a classical problem in graph theory that can be solved efficiently.  Applications of $st$-orientations include planarity testing (see, e.g.,~\cite[Chapter 2.3]{weiskircher01}), graph drawing (e.g.~\cite{biedl98}) and phylogenetics~\cite[Lemma 4.13]{janssen18}. 
For many applications and algorithms however, constrained acyclic orientations are of interest, e.g., acyclic orientations without any transitive arcs. In this context, an arc $(u,v)$ of a digraph is called \emph{transitive} if there exists a directed path from $u$ to $v$ that does not traverse $(u,v)$.

The problem of deciding if an undirected graph has an acyclic orientation without any transitive arc was first posed by Ore~\cite{ore62}. Subsequently, NP-completeness of the problem was established by Ne\v{s}et\v{r}il and R\"odel~\cite{nestril95} with an alternative and simpler proof given by Brightwell~\cite{brightwell93}. Complementing and simplifying results by R\"odel and Thoma~\cite{rodel95}, Binucci, Didimo, and Patrignani~\cite{binucci22} recently showed that the decision problem \NTstO, which asks if an undirected graph $G$ with vertex set $V$ and two vertices $s,t\in V$ has a non-transitive $st$-orientation, is NP-complete. In contrast to the earlier mentioned hardness results, note that the input of an instance of \NTstO\ consists not only of $G$, but also of two vertices $s$ and $t$ of $G$ that are the designated source and sink, respectively.  
In a separate paper, Binucci et al.~\cite{binucci23} added to the results presented in~\cite{binucci22} and showed that \NTstO\ is fixed-parameter tractable with respect to the treewidth of the input graph. Moreover they established NP-completeness of two special cases of the same problem: (i) $G$ has degree at most four and is a subdivision of a triconnected graph and (ii) $G$ has diameter at most six. In the final section of their paper, Binucci et al.\ provide several open problems one of which asks about the computational complexity of \NTstO\ if $G$ has degree at most three. By building on the construction in~\cite{binucci22} and using a novel gadget, we show in the first part of this paper  that \NTstO\ remains NP-complete on cubic graphs, i.e., for graphs with degree \textit{exactly} three.

In the second part of the paper, we then turn to an orientation problem that arises in the study of phylogenetic (evolutionary) networks. An \textit{unrooted phylogenetic network} $U$ is an undirected graph whose leaves are labeled with biological entities such as species, viruses, or individuals of a population. Due to their undirected nature, unrooted phylogenetic networks are limited in their ability to model ancestor-descendant relationships.
More precisely, an edge $\{u,v\}$ between two non-leaf vertices of $U$ does not provide an answer as to whether $u$ is an ancestor or descendant of $v$. On the other hand, a \textit{rooted phylogenetic network} $N$, which  is a leaf-labeled acyclic digraph that has a single source with out-degree at least two (the \textit{root}), each leaf is a sink, and all other vertices have in-degree one and out-degree at least two, or in-degree at least two and out-degree one, provides this information. If there is a directed path from $u$ to $v$ in $N$, then $u$ is an ancestor of $v$ and, biologically speaking, the older species represented by $u$ gave rise to the newer species represented by $v$ through vertical and horizontal descent such as speciation and hybridization.  
Initiated by Huber et al.~\cite{huber22}, studying the relationships between unrooted and rooted phylogenetic networks has recently become an active area of research in phylogenetics. A fundamental problem in this area is the following: Given an unrooted phylogenetic network $U$, does there exist a rooted phylogenetic network $N$ such that $U$ can be obtained from $N$ by suppressing its root and omitting the direction of each arc. 
Like deciding if an undirected graph has an $st$-orientation, this problem can be solved efficiently if there are no constraints on $N$~\cite[Lemma 4.13]{janssen18}.

Of particular interest in phylogenetics are unrooted (resp. rooted) {\it binary} phylogenetic networks whose internal vertices have degree three (resp. in-degree one and out-degree two, or in-degree two and out-degree one) because vertices of higher degree indicate insufficiencies in the dataset (e.g., DNA sequences) that was used for the phylogenetic network reconstruction. In addition to considering binary phylogenetic networks and due to the infinite number of phylogenetic networks that exist, even on just a single leaf, it is common to restrict the search space in the reconstruction of a phylogenetic network that best fits a dataset or is optimal with respect to some objective function to smaller classes of networks. For a recent summary of the most prominent classes of phylogenetic networks and their structural constraints, we refer the interested reader to Kong et al.~\cite{kong22}.

Returning back to the orientation of  phylogenetic networks with a focus on binary networks, Huber al al.~\cite{huber22}  investigated if an unrooted binary phylogenetic network can be oriented as a rooted binary phylogenetic network that belongs to certain classes of networks and, among other results, established a fixed-parameter tractable algorithm for a class of rooted binary phylogenetic networks that contains the popular class of binary tree-child networks. 
Since the computational complexity of the underlying problem was left open, their paper initiated a line of work on the hardness of deciding if an unrooted phylogenetic network can be oriented as a rooted phylogenetic network of some given class, which has since become an active area of research (e.g.,~\cite{bulteau23, dempsey24, doecker25, garvardt23, iersel,maeda23,urata24}). In the second part of this paper, we add to the literature of orienting phylogenetic networks and use our complexity result on non-transitive $st$-orientations for cubic graphs to show that it is NP-complete to decide if an unrooted binary phylogenetic network can be oriented as a rooted binary phylogenetic network without any transitive arc. Transitive arcs, which are called {\it shortcuts} in phylogenetics, play a role in the class of rooted {\it normal} networks~\cite{francis25,willson10}. This class contains every rooted phylogenetic network that has no transitive arc and for which each non-leaf vertex has a child that has in-degree one.

The remainder of the paper is structured as follows. In Section~\ref{sec:preliminaries}, we provide notation and definitions, and formally define \NTstO. Then, in Section~\ref{sec:new-fork-gadget}, we introduce a new gadget that is crucial to establish NP-hardness of \NTstO\ for cubic graphs via a Karp reduction from a variant of {\sc Not-All-Equal $3$-SAT} in Section~\ref{sec:nt-st-orientation-cubic}. Using the result of this section, we then show in Section~\ref{sec:phylogenetic-networks} that it is NP-complete to decide if an  unrooted binary phylogenetic network can be oriented as a rooted binary phylogenetic network without any transitive arc. Finally, Section~\ref{sec:conclusion} concludes the paper with a summary of the main results and some questions for future research.

\section{Preliminaries}\label{sec:preliminaries}

Let $G = (V, E)$ be a \textit{simple undirected graph} on a set $V$ of vertices and a set $E \subseteq \{\{u,v\} \mid u,v \in V \text{ with } u \neq v\}$ of edges. 
Let $V'$ be a subset of $V$. The simple undirected graph $G' = (V', E')$ with $E' = \{\{u,v\} \in E \mid u,v \in V'\}$ is called an \textit{induced subgraph} of $G$ and is denoted by $G[V']$. We say that two simple undirected graphs $G_1=(V_1, E_1)$ and $G_2=(V_2,E_2)$ are \textit{isomorphic} if and only if there exists a bijective mapping $f\colon V_1 \rightarrow V_2$ such that $\{u,v\} \in E_1$ if and only if $\{f(u),f(v)\} \in E_2$.  Next, let $D = (V, A)$ be a \textit{simple directed graph} (also called {\it simple digraph}) on a set $V$ of vertices and a set $A \subseteq \{(u,v) \mid u,v \in V \text{ with } u \neq v\}$ of arcs.
As all undirected (resp. directed) graphs in this paper are simple (i.e, they neither contain a loop nor a pair of edges (resp. arcs) in parallel), we will  omit the adjective {\it simple} throughout the paper. 

Now let $G=(V,E)$ be an undirected graph, and let $D=(V',A)$ be a directed graph. 
For an arc $(u,v)$ in $A$, we call  $u$  a \textit{parent} of $v$ and $v$  a \textit{child} of $u$. Furthermore, two vertices $u$ and $v$ are \textit{neighbors} in $D$ (resp.\ $G$) if $(u,v)\in A$ or $(v,u)\in A$ (resp.\ $\{u,v\}\in E$). We also say that two neighbors $u$ and $v$ are \textit{adjacent} to each other and that an edge $e\in E$ (resp.\ an arc $a\in A$) is {\it incident} with a vertex $v$ if  $v$ is one of the two endpoints of $e$ (resp.\ $a$). The \textit{degree} of a vertex $v$ in $G$ is the number of edges in $E$ that are incident with $v$. We call $G$ \textit{cubic} if each vertex in $G$ has degree three. Turning to $D$, let $w\in V'$. Then the \textit{in-degree} of $w$ in $D$ is $|\{(u,v) \in A \mid v = w\}|$ and the  \textit{out-degree}  of $w$ is $|\{(u,v) \in A \mid u = w\}|$. We also define the {\it degree} of $w$ in $D$ as the sum of its in-degree and out-degree. Moreover,  $w$  is called a \textit{source} of $D$ if its in-degree is zero and a \textit{sink} of $D$ if its out-degree is zero. Lastly, if a vertex of $G$ or $D$ has degree $k$, we refer to it as a \textit{degree-$k$ vertex}.

Next let $e_1 = \{u,v\}$ and $e_2 = \{v,w\}$ be two edges of an undirected graph $G = (V,E)$ such that $v$ is a degree-2 vertex. Further, let $G' = (V',E')$ be the undirected graph with $V' = V\setminus\{v\}$ and $E' = (E\setminus\{e_1, e_2\}) \cup \{\{u,w\}\}$, that is, $G$ is obtained from $G$ by deleting the two edges $e_1$ and $e_2$ and, if not already present, adding the edge $\{u,v\}$. We say that $G'$ is obtained from $G$ by \textit{suppressing} $v$. Reversely, $G$ is obtained from $G'$ by \textit{subdividing} the edge $\{u,w\}$ with a new vertex $v \not\in V'$. 
Furthermore, if there is a sequence $(G'=G_1, G_2, \ldots G_k=G)$ such that, for each $i \in \{1,2,\ldots,k-1\}$, $G_{i+1}$ can be obtained from $G_i$ by subdividing an edge, we say that $G$ is a \textit{subdivision} of $G'$. We also consider $G$ to be a subdivision of itself.  
      
Let $G=(V,E)$ be an undirected graph, and let $\pi = (v_1, v_2, \ldots, v_k)$ with $k\geq 1$ be a sequence of pairwise distinct vertices such that $e_i = \{v_i,v_{i+1}\}$ is an edge in $E$ for each $1 \leq i \leq k-1$. We say that $\pi$ is a \textit{path} of $G$.If $\pi$ is a path and there exists an edge $\{v_1,v_k\}$ in $G$, then we call $(v_1, v_2, \ldots, v_k,v_1)$ a \textit{cycle} of $G$. Since $G$ is simple, note that each cycle in $G$ has at least three distinct vertices.
If $G$ does not contain a cycle, then it is called \textit{acyclic}. Furthermore, we say that $G$ is \textit{connected} if, for each pair $u$ and $v$ of vertices in 
$V$, there exists a path in $G$ that connects $u$ and $v$. Now let $G=(V,E)$ be a connected undirected graph, and let $e\in E$. Then $e$ is referred to as  a \textit{cut edge} of  $G$ if the graph obtained from $G$ by deleting $e$ is no longer connected. Furthermore, a \textit{biconnected component} of $G$ is a maximal induced subgraph $G[V']$ such that $V'\subseteq V$ and, for any two distinct vertices $u,v \in V'$, there is a cycle in $G[V']$ that contains $u$ and $v$. Observe that any edge of $G$ that contains a vertex in $V'$ and a vertex in $V\setminus V'$ is a cut edge since $G[V']$ is maximal.      

Turning to directed graphs, let $D=(V,A)$ be a directed graph, and let $\pi = (v_1, v_2, \ldots, v_k)$ be a sequence of pairwise distinct vertices with $k\geq 1$ such that $a_i = (v_i,v_{i+1})$ is an arc in $A$ for each $1 \leq i \leq k-1$. We call $\pi$ a {\it directed path} in $D$.
Furthermore, we say that $\pi$ \textit{traverses} an arc $a \in A$ if $a = a_i$ for some $i \in \{1,2,\ldots,k-1\}$. If $\pi$ is a directed path and there exists an arc $(v_k,v_1)$ in $D$, then we call $(v_1, v_2, \ldots, v_k,v_1)$ a \textit{directed cycle} of $D$. If $D$ does not contain a directed cycle, then it is called \textit{acyclic}. 
Now, let $C=(v_1,v_2,\ldots,v_k)$ be a sequence of vertices vertices in $D$ such that the vertices in $\{v_1,v_2,\ldots,v_{k-1}\}$ are pairwise distinct and $v_1=v_k$. Then $C$ is an {\it underlying cycle} of $D$ if either $(v_i,v_{i+1})$ or $(v_{i+1}, v_i)$ is an arc in $D$  for each $1\leq i\leq k-1$.
Note that a directed cycle is always an underlying cycle, but that the reverse is not true in general. Intuitively, each cycle in the undirected graph $G'$ obtained from $D$ by omitting the direction of each arc is an underlying cycle of $D$. Observe however that $G'$ is not simple and may have a cycle that contains only two edges.
Lastly, an arc $(u,v)$ in $D$ is called a \textit{transitive arc} if  there exists a directed path from $u$ to $v$ in $D$ that does not traverse~$(u,v)$.

For the edge set $E$ of an undirected graph $G=(V,E)$, we define $A(E)$ to be be the set of arcs
\[
 \{(u,v),(v,u) \mid \{u,v\} \in E\}
\]
which contains both directions for each edge in $E$. Let $\tau \colon E \rightarrow A(E)$ be an injective function that assigns a direction to each edge of $G$, that is, $\tau(\{u,v\}) \in \{(u,v), (v,u)\}$. We refer to the operation of replacing an edge $e = \{u,v\}$ by an arc in $\{(u,v), (v,u)\}$ also as \textit{orienting} or {\it directing} $e$. We write $\cD_\tau(G)$ or, simply, $\cD(G)$ if $\tau$ does not play any particular role to denote the directed graph obtained from $G$ by replacing $e$ by the arc $\tau(e)$ for each $e\in E$ and say that $\cD_\tau(G)$ is an \textit{orientation} of $G$.  The focus of this paper are constrained orientations. In particular, we say that an orientation $\cD(G)$ is \textit{acyclic} (resp.\ \textit{non-transitive}) if $\cD(G)$ does not contain a directed cycle (resp.\ $\cD(G)$ does not contain a transitive arc). 
Now let $s,t\in V$. An \textit{$st$-orientation} of $G$ is an acyclic orientation $\cD(G)$ such that $s$ is the unique source of $\cD(G)$ and $t$ is the unique sink of $\cD(G)$. 
Intuitively, every vertex in $V\setminus\{s,t\}$ has an arc directed towards it and another arc directed away from it in any $st$-orientation of $G$.

The main focus of this paper is the following decision problem.

\noindent\fbox{%
    \parbox{.97\textwidth}{%
\begin{prb}
\noindent\NTstO\\
{\bf Instance.} An undirected graph $G=(V,E)$ and two vertices $s,t \in V$.\\
{\bf Question.} Does $G$ have a non-transitive $st$-orientation? \end{prb}
}
}

\begin{figure}[t]
    \centering
    \includegraphics[width=.5\textwidth]{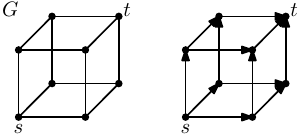}
    \caption{A cubic graph $G$ with designated vertices $s$ and $t$ (left) and a non-transitive $st$-orientation of $G$ (right).  }
    \label{fig:cube}
\end{figure}

\noindent The example illustrated in Figure~\ref{fig:cube} is a yes-instance of \NTstO. 

We next state three observations on non-transitive $st$-orientations whose proofs are straightforward and omitted. In particular, the third observation is an immediate consequence of the second.

\begin{observation}
\label{obs:symmetry}
Let $G=(V,E)$ be an undirected graph, let $s,t\in V$, and let $\cD(G)$ be an orientation of $G$. Then $\cD(G)$  is a non-transitive $st$-orientation of $G$ if and only if  the directed graph obtained from $\cD(G)$ by reversing the direction of each arc is a  non-transitive $ts$-orientation of $G$.
\end{observation}

\begin{observation}
\label{obs:two_arcs} 
Let $G=(V,E)$ be an undirected graph, let $s,t\in V$, and let $\cD(G)$ be a non-transitive $st$-orientation of $G$. Let $C$ be the set of arcs of an underlying cycle of $\cD(G)$. Then the direction of at least two arcs in $C$ has to be reversed such that the resulting arcs form a directed cycle.
\end{observation}

\begin{observation}\label{obs:triangle-free}
Let $G=(V,E)$ be an undirected graph with $s,t\in V$. If $G$ has a non-transitive $st$-orientation, then $G$ does not contain an induced subgraph isomorphic to a cycle with exactly three vertices.
\end{observation}

We end this section with notation for Boolean satisfiability problems that is needed in Section~\ref{sec:nt-st-orientation-cubic}  to show that {\NTstO} is NP-complete for cubic graphs. Let $X = \{x_1, x_2,\ldots,x_n\}$ be a set of Boolean \textit{variables}, that is, each $x_i \in X$ can take on the truth value \textit{true} (denoted by $T$) or the truth value \textit{false} (denoted by $F$). A \textit{literal} is a variable $x_i$ or its \textit{negation} $\overline{x_i}$, and the set of literals corresponding to $X$ is $\cL(X) = \{x_i, \overline{x_i} \mid x_i \in X\}$. A \textit{clause} $c_j$ is a subset of $\cL(X)$, and a \textit{formula} is a collection of clauses $C = \bigcup_{j=1}^m \{c_j\}$. A mapping $\beta\colon X \rightarrow \{T,F\}$ is called a \textit{truth assignment} for $X$. A truth assignment $\beta$ for $X$ is extended to a mapping $\beta_\cL\colon \cL(X) \rightarrow \{T,F\}$ as follows: $\beta_\cL(x_i) = \beta(x_i)$ and $\beta_\cL(\overline{x_i})=T$ if and only if $\beta(x_i) = F$.  We say that $\beta$ \textit{satisfies} a literal $l \in \cL(X)$ if  $\beta_\cL(l) = T$. Further, $\beta$ \textit{not-all-equal satisfies} a clause $c_j$, or short \textit{nae-satisfies} $c_j$,  if $\beta_\cL$ sets at least one literal in $c_j$ to $T$ and at least one literal in $c_j$ to~$F$. Finally, $\beta$ nae-satisfies a formula $C$ if $\beta$ nae-satisfies each clause $c_j \in C$, in which case we say that $C$ is \textit{nae-satisfiable}. If $|c_j| = 3$ for each clause in $C$, then the problem of deciding if $C$ is nae-satisfiable is referred to as \NAESAT. This problem
has been shown to be NP-complete by Schaefer~\cite{schaefer78}. 

\section{Fork gadget}\label{sec:new-fork-gadget}

In this section, we introduce the fork gadget that will play an important role in establishing NP-hardness of  {\NTstO} for cubic graphs. A variant of this gadget was introduced by Binucci et al.~\cite{binucci22} to prove NP-hardness of \NTstO\ for graphs with degree at most four. Intuitively, the fork gadget is an induced subgraph of an undirected graph $G$ that restricts non-transitive $st$-orientations of $G$ by linking the directions of three specific edges. While the fork gadget by Binucci et al. involves degree-1, degree-2, and degree-4 vertices, and cannot be used as an induced subgraph of a cubic graph, the fork gadget presented here only has degree-1 and degree-3 vertices.
\begin{figure}[t]
    \centering
    \includegraphics[width=.9\textwidth]{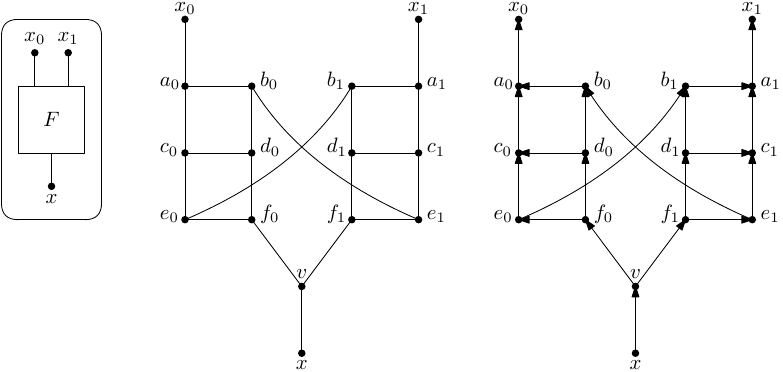}
    \caption{The fork gadget (middle), its schematic representation (left), and an acyclic and non-transitive orientation of the fork gadget (right). }
    \label{fig:fork-gadget-cubic}
\end{figure}

More formally, we refer to a graph $F$ that is isomorphic to the undirected graph depicted in Figure~\ref{fig:fork-gadget-cubic} as the {\it fork gadget}. The vertex set of $F$ is
\[
V_F = \{x,v\} \cup \{x_i,a_i,b_i,\ldots, f_i \mid i \in \{0,1\}\}.
\]
Let $G=(V,E)$ be an undirected graph, and let $s,t \in V$ such that $F$ is an induced subgraph of $G$ and $$\{s,t\} \cap \{a_i,b_i,\ldots, f_i \mid i \in \{0,1\}\} = \emptyset.$$ Further, let $\cD(G)$ be a non-transitive $st$-orientation of $G$. We will show in what follows that the fork gadget links the direction of the three edges $\{x,v\}$, $\{a_0,x_0\}$, and $\{a_1,x_1\}$. In particular, the construction of the fork gadget enforces that,
\begin{enumerate}[(i)]
\item if $(x,v)$ is an arc in $\cD(G)$, then $(a_0, x_0)$ and $(a_1,x_1)$ are arcs in $\cD(G)$, and 
\item if $(v,x)$ is an arc in $\cD(G)$, then $(x_0, a_0)$ and $(x_1,a_1)$ are arcs in $\cD(G)$.
\end{enumerate}
We remark that (i) and (ii) are analogous to the properties of the fork gadget by Binucci et al.~\cite{binucci22}.

To show that the fork gadget has the above two properties, we next establish two lemmas. The first lemma establishes properties of non-transitive $st$-orientations of an undirected graph that has an induced subgraph isomorphic to the \textit{ladder graph} $L_3$ with three rungs (see  Figure~\ref{fig:ladder-graph}). Observing that the  fork gadget has two induced vertex-disjoint subgraphs that are isomorphic to the $L_3$, we then use the first lemma to establishes properties of non-transitive $st$-orientations of an undirected graph that has an induced subgraph isomorphic to the fork gadget in the second lemma. The notation that is used in these two lemmas follows that of Figures~\ref{fig:fork-gadget-cubic} and~\ref{fig:ladder-graph}.
\begin{figure}[t]
    \centering
    \includegraphics[width=\textwidth]{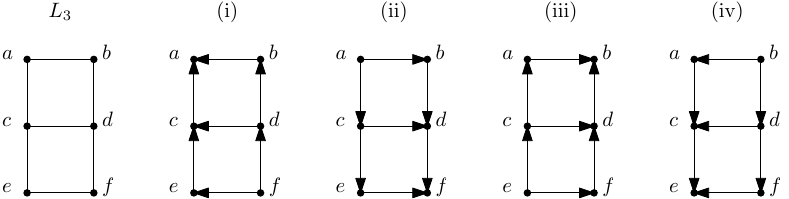}
    \caption{The ladder graph $L_3$ and four acyclic} and non-transitive orientations (i)--(iv) of $L_3$.
    \label{fig:ladder-graph}
\end{figure}

\begin{lemma}\label{lem:ladder}
Let $G=(V,E)$ be a connected undirected graph whose vertices have degree at most three, let $s,t\in V$, and let $\cD(G)$ be a non-transitive $st$-orientation of $G$. Further, let $V_L=\{a,b,\ldots, f\}$ be a subset of $V \setminus\{s,t\}$ such that the induced subgraph $G[V_L]$ is isomorphic to the ladder graph $L_3$. Then, $\cD(G)$ satisfies one of the following properties.
\begin{enumerate}[\rm (i)]
\item $(e,c), (c, a), (f, d), (d,b), (f, e), (d, c)$ and $(b, a)$ are arcs in $\cD(G)$.
\item $(c,e), (a,c), (d,f), (b,d), (e,f), (c,d)$ and $(a,b)$ are arcs in $\cD(G)$.
\item $(e,c), (c, a), (f, d), (d,b), (a, b), (c, d)$ and $(e, f)$ are arcs in $\cD(G)$.
\item $(c,e), (a,c), (d,f), (b,d), (b,a), (d, c)$ and $(f,e)$ are arcs in $\cD(G)$.
\end{enumerate}
\end{lemma}

\begin{proof}
Let $A$ denote the arc set of $\cD(G)$. To establish the lemma, we consider all possible orientations of the two edges $\{c, e\}$ and $\{c, d\}$. Throughout the proof, we freely use that $s,t \not \in V_L$ to infer arc directions.

Let us first consider the case $\{(e, c), (d, c)\} \subseteq A$. Then $c$ has in-degree two in $\cD(G)$ and, thus, $(c, a) \in A$. By Observation~\ref{obs:two_arcs}, it follows that $\{(d, b), (b, a)\} \subseteq A$. Then $d$ has out-degree two in $\cD(G)$ and, thus, $(f, d) \in A$. By Observation~\ref{obs:two_arcs}, we have $(f, e) \in A$ and, thus, (i) is satisfied. Next, let us consider  the case $\{(c, e), (c, d)\} \subseteq A$. Noting that the two arcs $(e, c)$ and $(d, c)$ considered in the first case determined the direction of each edge in $G[V_L]$, it now follows from Observation~\ref{obs:symmetry} that $\cD(G)$ satisfies (ii).

Now, if neither $\{(e, c), (d, c)\} \subseteq A$ nor $\{(c, e), (c, d)\} \subseteq A$, observe that there exists a directed path from $e$ to $d$ or from $d$ to $e$ in $\cD(G)$. Let us first consider the case $\{(e, c), (c, d)\} \subseteq A$. By Observation~\ref{obs:two_arcs}, it follows that $\{(e, f), (f, d)\} \subseteq A$. Then $d$ has in-degree two in $\cD(G)$ and, thus, $(d, b) \in A$. Again by Observation~\ref{obs:two_arcs}, we have $\{(c, a), (a, b)\} \subseteq A$ and, thus, (iii) is satisfied. Lastly, let us consider $\{(d, c), (c, e)\} \subseteq A$. Noting that the two arcs $(e, c)$ and $(c, d)$ considered in the last case determined the direction of each edge in $G[V_L]$, it now again follows from Observation~\ref{obs:symmetry} that $\cD(G)$ satisfies (iv). 

We conclude that $\cD(G)$ satisfies exactly one of the four properties (i)--(iv).    
\end{proof}

\begin{lemma}\label{lem:fork}
Let $G=(V,E)$ be a connected undirected graph whose vertices have degree at most three, let $s,t\in V$, and let $\cD(G)$ be a non-transitive $st$-orientation of $G$ with arc set $A$.
Let 
\[
V_F = \{x,x_0,x_1,v\} \cup V_0 \cup V_1 \text{ with } V_i = \{a_i, b_i, \ldots, f_i\} \text{ and } i\in\{0,1\}
\]
be a subset of $V$ such that $\{s,t\} \cap (V_0 \cup V_1 \cup \{v\}) = \emptyset$ and the induced subgraph $G[V_F]$ is isomorphic to the fork gadget. Then 
\[
\{(x,v),(a_0,x_0),(a_1,x_1)\} \subseteq A \text{ or } \{(v,x),(x_0,a_0),(x_1,a_1)\} \subseteq A.
\]
Furthermore, depending on the orientation of $\{v,x\}$ in $\cD(G)$, one of the  following properties is satisfied:
\begin{enumerate}[\rm (I)]
\item  If $(x,v) \in A$, then $v$ has in-degree one and out-degree two, and $a_0$ and $a_1$ have in-degree two and out-degree one in $\cD(G)$.
\item  If $(v,x) \in A$, then $v$ has in-degree two and out-degree one, and $a_0$ and $a_1$ have in-degree one and out-degree two in $\cD(G)$.
\end{enumerate}
\end{lemma}

\begin{proof}
\begin{figure}
    \centering
    \includegraphics[width=\textwidth]{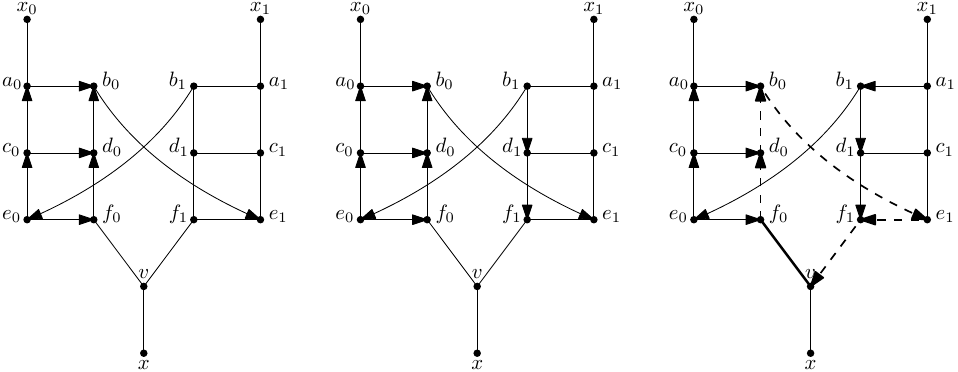}
    \caption{Left: The fork gadget in which the edges of $G[\{a_0,\ldots,f_0\}]$ are directed as in Figure~\ref{fig:ladder-graph}(iii) together with the two inferred arcs $(b_0,e_1)$ and $(b_1,e_0)$. Middle: The arcs $(b_1, d_1)$ and $(d_1, f_1)$ as inferred in the proof of Lemma~\ref{lem:fork}. Right: Directing $\{v, f_0\}$ results either in a transitive arc or a directed cycle with the dashed arcs.} 
    \label{fig:fork-cases-proof}
\end{figure}
For each $i \in \{0,1\}$, let 
\[
E_i = \{\{a_i, b_i\}, \{c_i, d_i\}, \{e_i, f_i\}, \{a_i, c_i\}, \{b_i, d_i\}, \{c_i, e_i\}, \{d_i, f_i\}\}
\]
be the edge set of the induced subgraph $G[V_i]$ that is isomorphic to the ladder graph $L_3$. By Lemma~\ref{lem:ladder}, there are four possibilities, which are illustrated in Figure~\ref{fig:ladder-graph}, to orient the edges in $E_i$ in obtaining $\cD(G)$: 
\begin{enumerate}[(i)]
\item $\{(e_i,c_i), (c_i, a_i), (f_i, d_i), (d_i,b_i), (f_i, e_i), (d_i, c_i), (b_i, a_i)\}$, 
\item $ \{(c_i,e_i), (a_i, c_i), (d_i, f_i), (b_i,d_i), (e_i, f_i), (c_i, d_i), (a_i, b_i)\}$,  
\item $ \{(e_i,c_i), (c_i, a_i), (f_i, d_i), (d_i,b_i), (a_i, b_i), (c_i, d_i), (e_i, f_i)\}$, or 
\item $\{(c_i,e_i), (a_i, c_i), (d_i, f_i), (b_i,d_i), (b_i, a_i), (d_i, c_i), (f_i,e_i)\}$.
\end{enumerate}

First, we show that the edges in $E_0$ are not oriented as in (iii) or (iv). Assume towards a contradiction that the edges in $E_0$ are oriented as in (iii). Noting that $e_0$ and $b_0$ are neither a source nor a sink in the fork gadget, it follows that $\{(b_0, e_1), (b_1, e_0)\} \subseteq A$. This setup is illustrated on the left-hand side of Figure~\ref{fig:fork-cases-proof}. Now, since one of (i)--(iv) applies, we have either $\{(b_1, d_1), (d_1, f_1)\} \subseteq A$ or $\{(f_1, d_1), (d_1, b_1)\} \subseteq A$. If the latter applies then, together with the arcs $(b_1, e_0),\, (e_0, f_0),\, (f_0, d_0),\, (d_0, b_0)$ and $(b_0, e_1)$, we have a directed cycle if $(e_1,f_1)\in A$ or a transitive arc if $(f_1,e_1)\in A$. Thus, $\{(b_1, d_1), (d_1, f_1)\} \subseteq A$ (see middle of Figure~\ref{fig:fork-cases-proof}).
In turn, $b_1$ has out-degree two in $\cD(G)$ and, so, $(a_1, b_1) \in A$. Again by Lemma~\ref{lem:ladder}, this implies that $(e_1, f_1) \in A$. Then $f_1$ has in-degree two in $\cD(G)$ and, thus, $(f_1, v) \in A$. Now, if $(v, f_0)\in A$, then $$(v,f_0),(f_0,d_0),(d_0,b_0),(b_0,e_1),(e_1,f_1)$$ is a directed cycle and, otherwise, $(f_0,v)$ is a transitive arc in $\cD(G)$ (see right-hand side of Figure~\ref{fig:fork-cases-proof}). It follows that the edges in $E_0$ are not oriented as in (iii). Moreover, by Observation~\ref{obs:symmetry}, the edges in $E_0$ are not oriented as in (iv) either. 

\begin{figure}[t]
    \centering
    \includegraphics[width=.7\textwidth]{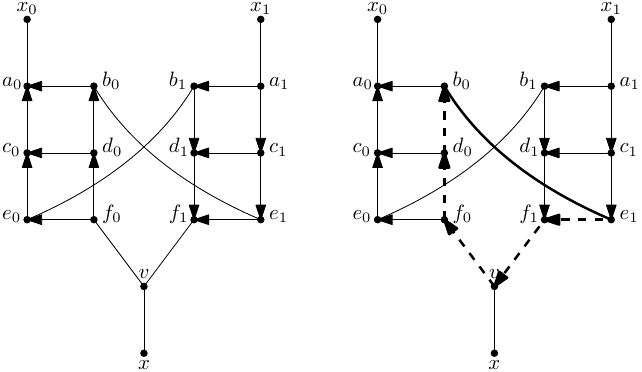}
    \caption{Left: The fork gadget in which the edges of $G[\{a_i,\ldots,f_i\}]$ are directed as in Figure~\ref{fig:ladder-graph}(i) for $i = 0$ and as in Figure~\ref{fig:ladder-graph}(ii) for $i=1$. Right: The two arcs $(f_1, v)$ and $(v,f_0)$ inferred in the proof of Lemma~\ref{lem:fork} such that orienting $\{b_0, e_1\}$ results either in a transitive arc or a directed cycle with the dashed arcs.} 
    \label{fig:for-cases-1-with-2}
\end{figure} 
Second, assume that the edges in $E_0$ are oriented as in (i). By symmetry of the fork gadget, it follows from the previous paragraph that the edges in $E_1$ are oriented as in (i) or (ii). If  (ii) applies, we obtain the arcs as illustrated on the left-hand side of Figure~\ref{fig:for-cases-1-with-2}. Then $f_0$ has out-degree two in $\cD(G)$ and $f_1$ has in-degree two in $\cD(G)$. Thus, $\{(f_1, v), (v,f_0)\} \subseteq A$. But now we cannot orient $\{b_0, e_1\}$ without introducing a transitive arc or a directed cycle (see right-hand side of Figure~\ref{fig:for-cases-1-with-2}). Hence, the edges in $E_1$ are oriented as in (i). Then $a_0$ and $a_1$ have in-degree two in $\cD(G)$, and $f_0$ and $f_1$ have out-degree two in $\cD(G)$. Thus, 
\[
\{(a_0, x_0), (a_1,x_1), (v, f_0), (v, f_1)\} \subseteq A
\]
and, as $v$ now has out-degree two, $(x,v) \in A$. In summary, we have $\{(x,v),(a_0,x_0),(a_1,x_1)\} \subseteq A$, $v$ has in-degree one and out-degree two, and $a_0$ and $a_1$ have in-degree two and out-degree one in $\cD(G)$.

Third, assume that the edges in $E_0$ are oriented as in (ii). By symmetry of the fork gadget, it  follows from the previous  two paragraphs that the edges in $E_1$ are oriented as in~(ii). Then $f_0$ and $f_1$ have in-degree two in $\cD(G)$, whereas $a_0$ and $a_1$ have out-degree two. Therefore, we have $\{(x_0, a_0), (x_1,a_1), (f_0,v), (f_1,v)\} \subseteq A$ and, as $v$ now has in-degree two, we also infer $(v,x) \in A$. In summary, we have $\{(v,x),(x_0,a_0),(x_1,a_1)\} \subseteq A$,  $v$ has in-degree two and out-degree one, and $a_0$ and $a_1$ have in-degree one and out-degree two in $\cD(G)$. This completes the proof of the lemma.
\end{proof}

\begin{figure}[t]
    \centering
    \includegraphics[width=.8\textwidth]{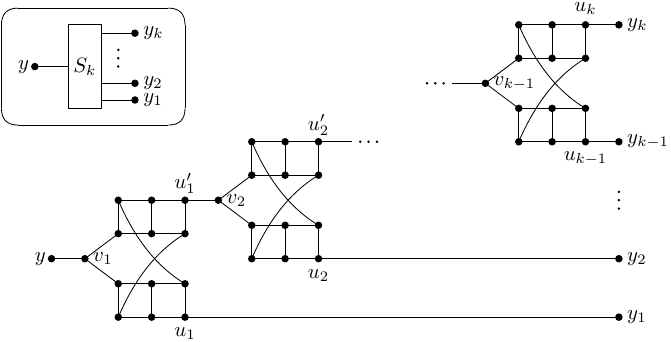}
    \caption{The split gadget $S_k$ and its schematic representation (top left).}
    \label{fig:split-gadget-cubic}
\end{figure}

We end this section by combining several copies of the fork gadget into a new gadget which we call the {\it split gadget} $S_k$. While the fork gadget links the orientation of three edges or, in other words,  enforces directions on the two edges $\{a_0,x_0\}$ and $\{a_1,x_1\}$ given the direction of the edge $\{v,x\}$, the split gadget enforces  directions on $k \geq 2$ edges given the direction of a single edge. In particular, for $k=2$, the split gadget is isomorphic to the fork gadget. For any integer $k\geq 2$, Figure~\ref{fig:split-gadget-cubic} illustrates  the split gadget and its schematic representation that will be useful in Section~\ref{sec:nt-st-orientation-cubic} when the split gadget is  part of a larger constructions. We note that the split gadget is the counterpart of the split gadget as presented in~\cite[Figure 5]{binucci22} and follows the same ideas.

The next lemma is the analog of Lemma~\ref{lem:fork} for the split gadget. The notation that is used in the lemma follows that of Figure~\ref{fig:split-gadget-cubic}.

\begin{lemma}\label{lem:split}
Let $G=(V,E)$ be a connected undirected graph whose vertices have degree at most three, let $s,t\in V$, and let $\cD(G)$ be a non-transitive $st$-orientation of $G$ with arc set $A$. Further, let $V'$ be a subset of $V$ such that the induced subgraph $G[V']$ is isomorphic to the split gadget $S_k$ for some $k\geq 2$ such that no degree-3 vertex in $G[V']$ is $s$ or $t$. Then 
\begin{align*}
&\{(y,v_1),(u_1,y_1),(u_2,y_2), \ldots, (u_k, y_k)\} \subseteq A \text{ or }\\
&\{(v_1,y),(y_1,u_1),(y_2,u_2), \ldots, (y_k, u_k)\} \subseteq A.
\end{align*}
Furthermore, depending on the orientation of $\{v_1,y\}$ in $\cD(G)$, one of the following properties holds.
\begin{enumerate}[\rm (I)]
\item If $(y,v_1) \in A$, then $v_1$ has in-degree one and out-degree two and each of $u_1,u_2,\ldots u_k$ has in-degree two and out-degree one in $\cD(G)$.
\item If $(v_1,y) \in A$, then $v_1$ has in-degree two and out-degree one and each of  $u_1,u_2,\ldots u_k$ has in-degree one and out-degree two in $\cD(G)$.
\end{enumerate}
\end{lemma}
\begin{proof}
First assume that $(y,v_1) \in A$. Then $\{(u_1,y_1), (u'_1, v_2)\} \subseteq A$ by Lemma~\ref{lem:fork}. Now, as $(u'_1, v_2) \in A$, it follows again by Lemma~\ref{lem:fork} that $\{(u_2,y_2), (u'_2, v_3)\} \subseteq A$. Repeating this argument for $$(u'_2, v_3), (u'_3,v_4),\ldots, (u'_{k-2}, v_{k-1})$$ establishes $\{(y,v_1),(u_1,y_1),(u_2,y_2), \ldots, (u_k, y_k)\} \subseteq A$. 
Moreover, as $(y, v_1) \in A$ and $(u_{i-1}',v_i) \in A$ for each $2 \leq i \leq k-1$, it again follows from repeated applications of Lemma~\ref{lem:fork}  that $v_1$ has in-degree one and out-degree two and each of $u_1,u_2,\ldots u_k$ has in-degree two and out-degree one. 

Second assume that $(v_1,y) \in A$. Then, it follows by arguments that are analogous to those used in the first paragraph that $$\{(v_1,y),(y_1,u_1),(y_2,u_2), \ldots, (y_k, u_k)\} \subseteq~A,$$ $v_1$ has in-degree two and out-degree one, and each of $u_1,u_2,\ldots u_k$ has in-degree one and out-degree two. 
\end{proof}

\section{Non-transitive st-orientations of cubic graphs}\label{sec:nt-st-orientation-cubic}

In this section, we show that \NTstO\ is NP-complete on cubic graphs, via a polynomial-time Karp reduction from a restricted variant of \NAESAT. Although our overall approach is similar to that taken by Binucci et al.~\cite{binucci22}, the new fork gadget, as introduced in the last section, is key in establishing our result.

We start by building some intuition about the construction of an undirected graph $G$ with two designated vertices $s$ and $t$ from a collection of clauses $C$ over a set of variables $X$. A clause $c_j$ is represented by a single vertex $c_j \not \in \{s,t\}$ in $G$ such that its incident edges correspond to the three literals that the clause $c_j$ contains. For any $st$-orientation~$\cD(G)$, the orientation of these edges correspond to truth values of the represented literals. More specifically, an orientation of an edge towards $c_j$ corresponds to a true literal and an orientation of an edges away from $c_j$ corresponds to a false literal, Observe that $c_j$ has in-degree and out-degree at least one in $\cD(G)$ since $c_j\notin \{s,t\}$. 
The remaining parts of $G$ are constructed by combining multiple copies of the fork gadget such that
\begin{enumerate}[(i)]
\item $\cD(G)$ correctly simulates a truth assignment for $X$, that is, the orientations of edges representing a literal are consistent across clause vertices, and   
\item $G$ has a non-transitive $st$-orientation if and only if $C$ is nae-satisfiable. 
\end{enumerate}

In contrast to Binucci et al.~\cite{binucci22}, who reduced an instance of \NAESAT\ to an undirected graph that contains degree-2, degree-3, and degree-4 vertices, our construction results in an undirected graph that contains degree-2 and degree-3 vertices only. As we will later see, the degree-2 vertices can safely be suppressed so that we eventually obtain a cubic graph. Avoiding degree-2 and degree-4 vertices required significant changes  to the construction and arguments presented in~\cite{binucci22}. The central building block in achieving this goal is the fork gadget as presented in the last section. In addition to this new gadget, we also reduce from a variant of  \NAESAT\ in which each literal is a variable and each variable appears exactly four times in a collection of clauses. These constraints substantially simplify the upcoming construction and make the arguments more concise, but are not strictly necessary to show that \NTstO\ is NP-complete for cubic graphs.

We now formally state the variant of  \NAESAT\ that we use to establish our NP-completeness result. 

\noindent\fbox{%
    \parbox{.97\textwidth}{%
\begin{prb}
\noindent\RNAESAT\\
{\bf Instance.} A set of variables $X$ and collection of clauses $C$ such that 
\begin{enumerate}[(i)]
\item for each $c_j \in C$, we have $c_j \subseteq X$ and $|c_j| = 3$, and
\item each variable $x_i \in X$ appears in exactly four clauses of $C$.
\end{enumerate}
{\bf Question.} Is there a truth assignment $\beta\colon X \rightarrow \{T,F\}$ that nae-satisfies $C$? \end{prb}
}
}

\noindent It was shown in~\cite[Theorem.~1]{darmann20} that \RNAESAT\ is NP-complete.

\paragraph*{\textbf{Example.}}
Let $X = \{x_1, x_2, \ldots, x_6\}$ be a set of variables, and let
\begin{align*}
C_{\text{yes}}= \{ &\{x_1, x_2, x_3\}, \{x_1, x_2, x_4\}, \{x_1, x_2, x_5\}, \{x_1, x_2, x_6\}, \\
       &\{x_3, x_4, x_5\}, \{x_3, x_4, x_6\}, \{x_3, x_5, x_6\}, \{x_4, x_5, x_6\}\}
\end{align*}
be a collection of clauses over $X$. Then, $\cI = (X, C_{\text{yes}})$ is a yes-instance of \RNAESAT\ as the truth assignment $\beta\colon X \rightarrow \{T, F\}$ with $\beta(x) = T$ for $x \in \{x_1,x_3,x_4\}$ and $\beta(x) = F$ for $x \in \{x_2,x_5,x_6\}$ nae-satisfies each clause. 

The remainder of the section is organized as follows. We next introduce three additional gadgets: the  source gadget, the sink gadget, and the variable gadget whose central building blocks are copies of the fork gadget. As the names suggest,  the source and sink gadget contains a source and sink vertex, respectively, when its edges are oriented. For each of the source, sink, and variable gadget, we then establish  properties of the orientation of its edges if the gadget is an induced subgraph of an undirected graph that has a non-transitive $st$-orientation (Lemmas~\ref{lem:source}--\ref{lem:variable}). Subsequently, we consider two undirected graphs $G$ and $G'$ such that $G'$ is a subdivision of $G$ and analyze which properties must be satisfied such that the existence of a non-transitive $st$-orientation in one of $G$ and $G'$ guarantees the existence of such an orientation in the other graph (Lemmas~\ref{lem:subdivision} and~\ref{lem:suppress-W}). Lastly, we establish the main result of this section, showing that \NTstO\ is NP-complete for cubic graphs (Theorem~\ref{thm:main-nt-st-orientation}).

\begin{figure}[t]
    \centering
    \includegraphics[width=\textwidth]{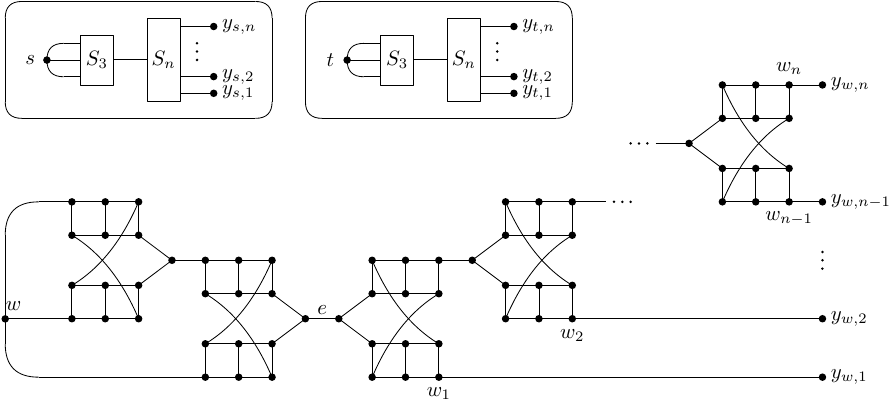}
    \caption{The source gadget with $w=s$ and $w_i=s_i$ for all $i\in\{1,2\ldots,n\}$, and the sink gadget with $w=t$ and $w_i=t_i$ for all $i\in\{1,2\ldots,n\}$, where $n$ is the number of variables of an instance of \RNAESAT. The schematic representations of the source and sink gadgets are shown at the top.}
    \label{fig:sosi-gadget-cubic}
\end{figure}
\paragraph*{\textbf{The source gadget and sink gadgets.}} Let $\cG(s)$ denote the {\it source gadget}, and let $\cG(t)$ denote the {\it sink gadget} as illustrated in Figure~\ref{fig:sosi-gadget-cubic}. Observe that $\cG(s)$ and $\cG(t)$ are isomorphic. Essentially, each of $\cG(s)$ and $\cG(t)$ consists of the two split gadgets $S_3$ and $S_n$ that are connected via a cut edge. Intuitively, for an instance of \RNAESAT\  with $n$ variables, the source gadget simulates a source vertex with degree $n$ by enforcing the directions of $n$ specific edges in any acyclic and non-transitive orientation such that $s$ is the unique source. The sink gadget has an analogous intuition. The notation that is used in the next two lemmas follows that of Figure~\ref{fig:sosi-gadget-cubic}.

\begin{lemma}\label{lem:source}
Let $G=(V,E)$ be a connected undirected graph whose vertices have degree at most three, let $s,t\in V$, and let $\cD(G)$ be a non-transitive $st$-orientation of $G$ with arc set $A$. Further, let $V'$ be a subset of $V\setminus\{t\}$ such that the induced subgraph $G[V']$ is isomorphic to $\cG(s)$. Then the following statements hold:
\begin{enumerate}[\rm (i)]
\item there exists an acyclic and non-transitive orientation of $G[V']$ such that $s$ is the unique source and the vertices $y_{s,1}, y_{s,2}, \ldots, y_{s,n}$ are precisely the sinks, 
\item $\{(s_1,y_{s,1}), (s_2,y_{s,2}), \ldots, (s_n,y_{s,n})\} \subseteq A$, and
\item  the vertices $s_1, s_2, \ldots, s_n$ have in-degree two and out-degree one in~$\cD(G)$.
\end{enumerate}
\end{lemma}
\begin{proof}
\begin{figure}[t]
    \centering
    \includegraphics[width=\textwidth]{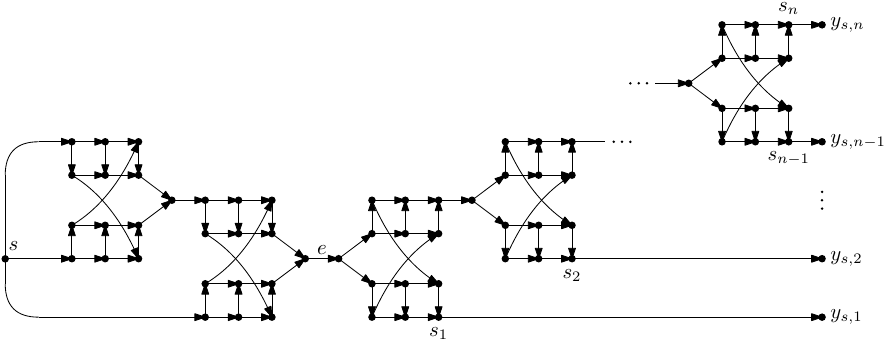}
    \caption{An acyclic and non-transitive orientation of the source gadget with a single source $s$ and sinks $y_{s,i}$ with $1 \leq i \leq n$. Reversing each arc direction yields an acyclic and non-transitive orientation of the sink gadget with a single sink $t$ and sources $y_{t,i}$ for $1 \leq i \leq n$, and where $t = s$ and $t_i = s_i$.}
    \label{fig:sosi-gadget-cubic-orientation}
\end{figure}
The orientation of the source gadget that is illustrated in Figure~\ref{fig:sosi-gadget-cubic-orientation} establishes (i).  Now, as $e$ is a cut edge of $G$ that is incident to a vertex of the biconnected component $B_s$ that contains $s$, it follows that $e$ is oriented away from $B_s$ in $\cD(G)$. Then, by Lemma~\ref{lem:split}, we have
\[
\{(s_1,y_{s,1}), (s_2,y_{s,2}), \ldots, (s_n,y_{s,n})\} \subseteq A
\]
such that $s_i$ has in-degree two and out-degree one for $1 \leq i \leq n$. This establishes (ii) and (iii).
\end{proof}

\begin{lemma}\label{lem:sink}
Let $G=(V,E)$ be a connected undirected graph whose vertices have degree at most three, let $s,t\in V$, and let $\cD(G)$ be a non-transitive $st$-orientation of $G$ with arc set $A$. Further, let $V'$ be a subset of $V\setminus\{s\}$ such that the induced subgraph $G[V']$ is isomorphic to $\cG(t)$. Then the following statements hold:
\begin{enumerate}[\rm (i)]
\item there exists an acyclic and non-transitive orientation of $G[V']$ such that $t$ is the unique sink and the vertices $y_{t,1}, y_{t,2}, \ldots, y_{t,n}$ are precisely the sources,
\item $\{(y_{t,1}, t_1), (y_{t,2},t_2), \ldots, (y_{t,n}, t_n)\} \subseteq A$,  and
\item the vertices $t_1, t_2, \ldots, t_n$ have in-degree one and out-degree two~in $\cD(G)$.
\end{enumerate}
\end{lemma}
\begin{proof}
Reversing the orientation of each arc in Figure~\ref{fig:sosi-gadget-cubic-orientation} establishes (i). Again as $e$ is a cut edge of $G$ that is incident to a vertex of the biconnected component $B_t$ that contains $t$, it follows that the arc $e$ is oriented towards $B_t$ in $\cD(G)$. Then, by Lemma~\ref{lem:split}, we have 
$$
\{(y_{t,1}, t_1), (y_{t,2},t_2), \ldots, (y_{t,n}, t_n)\} \subseteq A
$$
such that $t_i$ has in-degree one and out-degree two for $1 \leq i \leq n$. This establishes (ii) and (iii).
\end{proof}

\paragraph*{\textbf{The variable gadget.}} For a variable $x_i$, let $\cG(x_i)$ denote the {\it variable gadget} as illustrated in Figure~\ref{fig:variable-gadget-cubic}. The variable gadget consists of two copies of the fork gadget $F$ and one copy of the split gadget $S_4$ as  induced subgraphs. Here, we use $S_4$ since $x_i$ appears exactly four times in an instance of \RNAESAT. Then, following the same notation as in Figure~\ref{fig:variable-gadget-cubic}, the next lemma shows that any acyclic and non-transitive orientation of $\cG(x_i)$ that does not have a vertex of in-degree three or out-degree three has either the arcs $(u_k, x_{i,k})$ or the arcs $(x_{i,k}, u_k)$ for each $k \in \{1,2,3,4\}$. We refer to these orientations  as the \textit{true} configuration and the \textit{false} configuration of $\cG(x_i)$, respectively. For an illustration of these two configurations, see  Figure~\ref{fig:variable-gadget-cubic-orientations}.

\begin{figure}[t]
    \centering
    \includegraphics[width=.95\textwidth]{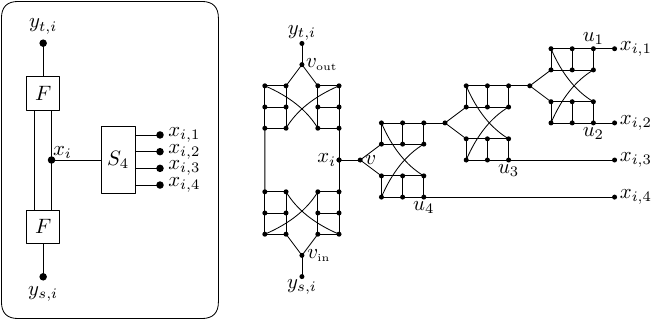}
    \caption{The variable gadget $\cG(x_i)$ and its schematic representation (left).}
    \label{fig:variable-gadget-cubic}
\end{figure}

\begin{lemma}\label{lem:variable}
Let $G=(V,E)$ be a connected undirected graph whose vertices have degree at most three, let $s,t\in V$, and let $\cD(G)$ be a non-transitive $st$-orientation of $G$ with arc set $A$. Further, let $V'$ be a subset of $V\setminus\{s,t\}$ such that the induced subgraph $G[V']$ is isomorphic to $\cG(x_i)$. Then the following statements hold:
\begin{enumerate}[\rm (i)]
\item there exists an acyclic and non-transitive orientation of $G[V']$ such that $y_{s,i}$ is the unique source and $y_{t,i}$, $x_{i,1},x_{i,2},x_{i,3}, x_{i,4}$ are precisely the sinks,
\item there exists an acyclic and non-transitive orientation of $G[V']$ such that $y_{s,i}$, $x_{i,1}$, $x_{i,2},x_{i,3}, x_{i,4}$ are precisely the sources and $y_{t,i}$ is the unique sink, 
\item$\{(x_i, v)\} \cup \{(u_k, x_{i,k}) \mid 1 \leq k \leq 4\} \subseteq A$ or $\{(v, x_i)\} \cup \{(x_{i,k}, u_k) \mid 1 \leq k \leq 4\} \subseteq A$,
\item if $(x_i,v) \in A$, then $x_{i,k}$ is the unique child of $u_k$ for $1\leq k \leq 4$, and
\item if $(v,x_i) \in A$, then $x_{i,k}$ is the unique parent of $u_k$ for  $1\leq k \leq 4$.
\end{enumerate}
\end{lemma}
\begin{proof}
The two orientations of  $\cG(x_i)$ that are illustrated in Figure~\ref{fig:variable-gadget-cubic-orientations} establish (i) and (ii). More precisely, if $\cG(x_i)$ is in the {\it true} (resp. {\it false}) configuration, then (i) (resp. (ii)) holds.
Moreover,  (iii)--(v) follow from Lemma~\ref{lem:split}. 
\end{proof}

\begin{figure}[t]
    \centering
    \includegraphics[width=.85\textwidth]{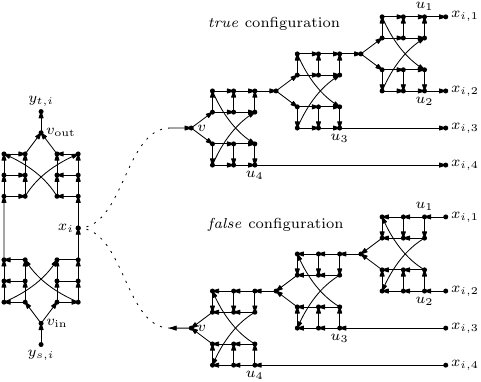}
    \caption{Two orientations of the variable gadget $\cG(x_i)$ corresponding to the truth assignment $\beta(x_i) = T$ (\textit{true} configuration) and  $\beta(x_i) = F$ (\textit{false} configuration). The part of $\cG(x_i)$ that is shown on the left and consists of two copies of the fork gadget has the same arc directions in both configurations.}
    \label{fig:variable-gadget-cubic-orientations}
\end{figure}

\paragraph*{\textbf{Vertices of degree two}} The following two lemmas allow us to connect several gadgets by identifying their respective degree-1 vertices and subsequently suppressing the resulting degree two vertices without compromising the existence of a non-transitive $st$-orientation in the resulting undirected graph.

\begin{lemma}\label{lem:subdivision}
Let $G = (V,E)$ be an undirected graph with $s,t\in V$, and  let $G'$ be a subdivision of $G$. If $G$ has a non-transitive $st$-orientation, then $G'$ has a non-transitive $st$-orientation. 
\end{lemma} 
\begin{proof}
Suppose that $G$ has a non-transitive $st$-orientation $\cD_\tau(G)$. For each edge $e \in E$, let
\[
P_e = (v_1, v_2, \ldots, v_k) \text{ with } e = \{v_1, v_k\}
\]
be the path in $G'$ such that $v_i$ has degree two in $G'$ for each $i \in \{2,3,\ldots,k-1\}$. If $k = 2$, then $P_e$ is also an edge in $G'$ and, thus, $P_e$ consists of the two degree three vertices $v_1$ and $v_k$. We now obtain a non-transitive $st$-orientation $\cD_{\tau'}(G')$ as follows. If $\tau(\{v_1, v_k\}) = (v_1, v_k)$, we set $\tau'(\{v_i, v_{i+1}\}) = (v_i, v_{i+1})$ for $1 \leq i \leq k-1$ and, if $\tau(\{v_1, v_k\}) = (v_k, v_1)$, we set $\tau'(\{v_i, v_{i+1}\}) = (v_{i+1}, v_i)$ for $1 \leq i \leq k-1$. As  $\cD_{\tau}(G)$ is a non-transitive $st$-orientation of $G$, it is now straightforward to verify that $\cD_{\tau'}(G')$ is such an orientation of $G'$.
\end{proof}

The converse of Lemma~\ref{lem:subdivision} does not hold in general. However, the following lemma gives sufficient conditions for a subdivision such that the existence of a non-transitive $st$-orientation is preserved when suppressing degree-2 vertices. 
 
\begin{lemma}\label{lem:suppress-W}
Let $G = (V,E)$ be an undirected graph with $s,t\in V$ and such that each vertex has degree two or three. Suppose that $G$ has a non-transitive $st$-orientation. Furthermore, let $W$ be a subset of  $V\setminus\{s,t\}$ such that the following three properties are satisfied for each $w\in W$:
\begin{enumerate}[\rm (i)]
\item $w$ has degree two,
\item both neighbors $u$ and $v$ of $w$ have degree three in $G$, and 
\item every non-transitive $st$-orientation of $G$ has exactly one directed path connecting $u$ and $v$.
\end{enumerate} 
Then, the undirected graph obtained from $G$ by suppressing each vertex in $W$ is simple and has a non-transitive $st$-orientation.     
\end{lemma} 

\begin{proof}
Let $\cD_\tau(G)$ be a non-transitive $st$-orientation of $G$. By the statement of the lemma,  $\cD_\tau(G)$  exists. We use induction on $|W|$ to establish the lemma. If $|W|=0$, the lemma clearly holds. Assume that $|W|=k$ and that the lemma holds for every proper subset of $W$. Let $w\in W$, and let $u$ and $v$ be the two neighbors of $w$ in $G$. Furthermore, let $G'=(V',E')$ be the graph obtained from $G$ by suppressing $w$, i.e., $$V'=V\setminus \{w\}\text{ and }E'=(E\setminus\{\{u,w\},\{w,v\}\})\cup\{\{u,v\}\}.$$
By Observation~\ref{obs:triangle-free}, $G$ does not have an induced subgraph that is isomorphic to a cycle with exactly three vertices and, thus, $G'$ is simple.
Guided by $\tau$, we next obtain a mapping $\tau'$ for the edges in $G'$. Let $\tau'\colon E'\rightarrow A(E')$ be the mapping defined as follows:
\begin{enumerate}[(a)]
\item if $\tau(\{u,w\}) = (u,w)$, then set $\tau'(\{u,v\}) = (u,v)$  and, otherwise, set $\tau'(\{u,v\}) = (v,u)$ and
\item for each edge $e$ in $E'\setminus\{\{u,v\}\}$, set $\tau'(e) = \tau(e)$.

\end{enumerate}

As $s$ and $t$ is the unique source and sink, respectively, in $G$, note that $w$ has in-degree one and out-degree one in $\cD_{\tau}(G)$. Moreover, since $\cD_\tau(G)$ is a non-transitive $st$-orientation of $G$, it is straightforward to check that $\cD_{\tau'}(G')$ is acyclic and, for each $e\in E'\setminus\{\{u,v\}\}$, the arc $\tau'(e)$ is not a transitive arc in $\cD_{\tau'}(G')$. Now assume towards a contradiction that $\tau'(\{u,v\})$ is a transitive arc in $\cD_{\tau'}(G')$. If $(u,v)$ (resp. $(v,u)$) is an arc in $\cD_{\tau'}(G')$, then there exists a directed path $\pi$ from $u$ to $v$ (resp. from $v$ to $u$) in $\cD_{\tau'}(G')$ that consists of at least two arcs. Moreover $\pi$ is also a directed path in $\cD_{\tau}(G)$.
It now follows that $\cD_\tau(G)$ contains two directed paths connecting $u$ and $v$; a contradiction. Hence $\tau'(\{u,v\})$ is not a transitive arc in $\cD_{\tau'}(G')$ and, in turn,  $\cD_{\tau'}(G')$ is a non-transitive $st$-orientation of $G'$. We next show that $G'$ satisfies (i)--(iii) relative to $W\setminus \{w\}$. Let $w'\in W\setminus\{w\}$, and let $u'$ and $v'$ be the two neighbors of $w'$. By construction of $G'$ from $G$, $w'$ clearly satisfies (i) and (ii). To see that (iii) also holds, assume again towards a contradiction that there exists a non-transitive $st$-orientation $\cD_{\tau''}(G')$ such that there is a directed path that connects $u'$ and $v'$ without traversing $w'$. If $(u,v)$ (resp. $(v,u)$) is an arc in $\cD_{\tau''}(G')$, then we obtain a non-transitive $st$-orientation for $G$ from $\cD_{\tau''}(G')$ by replacing $(u,v)$ (resp. $(v,u)$) with the two arcs $(u,w)$ and $(w,v)$ (resp. $(v,w)$ and $(w,u)$).  As $w'\in W$, this orientation for $G$ has, in particular,  the property that there are two directed paths connecting $u'$ and $v'$; another contradiction. Hence, $G'$ satisfies (i)--(iii) relative to $W\setminus \{w\}$. It now follows from the induction assumption that the graph obtained from $G'$ by suppressing each vertex in $W\setminus \{w\}$ is simple and has a non-transitive $st$-orientation. This concludes the proof of the lemma.
\end{proof} 

 \paragraph*{\textbf{The main result.}} 
 We are now in a position to establish the main result of this section, which necessitates some additional terminology. Let $C=\{c_1,c_2,\ldots,c_m\}$ be a collection of clauses over a set of variables $X=\{x_1,x_2,\ldots,x_n\}$. As we will refer to specific appearances of variables in the proof of the next theorem, we define, 
for each combination of $1\leq i\leq n$ and $1\leq j\leq m$ such that $x_i\in c_j$, $$\alpha(x_i, c_j) = k \text{ with } k\in\{1,2,3,4\}$$ if exactly $k-1$ clauses in $\{c_1,c_2,\ldots,c_{j-1}\}$ contain $x_i$.
Intuitively, for a clause $c_j \in C$ and a variable $x_i \in c_j$, the value $\alpha(x_i, c_j) = k$ indicates that $c_j$ contains the $k$-th appearance of $x_i$ in $C$.

\begin{theorem}\label{thm:main-nt-st-orientation}
\NTstO\ is {\rm NP}-complete for cubic graphs. 
\end{theorem} 

\begin{proof}
We start by showing that \NTstO\ is in NP. Let $\cD(G)$ be an orientation of an undirected graph $G=(V,E)$ with $s,t\in V$. By considering, in turn, the in-degree and out-degree of each vertex in $\cD(G)$, we can verify in polynomial time that $s$ is the unique source and $t$ is the unique sink of $\cD(G)$. Furthermore, verifying that $\cD(G)$is acyclic can also be done in polynomial time by using depth-first search.
Lastly, for each arc $(u,v)$ of $\cD(G)$, we can again use depth-first search to verify that there is no directed path from $u$ to $v$ in the directed graph obtained from $\cD(G)$ by deleting $(u,v)$. Hence,  we can check in time that is polynomial in $|V|$ if $\cD(G)$ is a non-transitive $st$-orientation of $G$.

We complete the proof of the theorem by establishing NP-hardness of \NTstO\ for cubic graphs by a reduction from \RNAESAT. Let $\cI = (X, C)$ be an instance of \RNAESAT\ with variables $X = \{x_1, x_2, \ldots, x_n\}$ and clauses $C = \{c_1, c_2, \ldots, c_m\}$. Throughout the remainder of the proof, we use $i$, $j$, and $k$ to index variables, clauses, and appearances of variables, respectively. More specifically $1\leq i\leq n$, $1\leq j\leq m$, and $1\leq k\leq 4$.

Let $\cG(s)$ be the source gadget, and let $\cG(t)$ be the sink gadget such that each of  $\cG(s)$ and $\cG(t)$ consists of the two split gadgets $S_3$ and $S_n$. In particular, the vertices $y_{s,i}$ of $\cG(s)$ (resp.\ $y_{t,i}$ of $\cG(t)$) are precisely the degree-1 vertices of $\cG(s)$ (resp.\ $\cG(t)$). For each variable $x_i \in X$, we apply the following  three steps one after the other to obtain a connected undirected graph $G'$ from $\cG(s)$ and $\cG(t)$.
\begin{enumerate}[(a)]
\item Introduce a new copy of the variable gadget~$\cG(x_i)$, where $y_{s,i}$, $y_{t,i}$, $x_{i,1}$, $x_{i,2}$, $x_{i,3}$ and $x_{i,4}$ are precisely the degree-1 vertices of $\cG(x_i)$.   
\item Identify $y_{s,i}$ of $\cG(s)$ with $y_{s,i}$ of~$\cG(x_i)$.
\item Identify $y_{t,i}$ of ~$\cG(t)$ with $y_{t,i}$ of~$\cG(x_i)$.
\end{enumerate}   
Next we obtain a connected undirected graph $G''$ from $G'$ as follows. For each clause $c_j = \{x_p, x_q, x_r\} $ with $p,q,r\in\{1,2,\ldots,n\}$, we introduce a new vertex $c_j$, referred to as a {\it clause gadget}, and three edges $\{x_{p,\kappa}, c_j\}$, $\{x_{q,\kappa'}, c_j\}$ and $\{x_{r,\kappa''}, c_j\}$, where $\kappa = \alpha(x_p, c_j)$, $\kappa' = \alpha(x_q, c_j)$ and $\kappa'' = \alpha(x_r, c_j)$. 
Observe that each degree-1 vertex $x_{i,k}$ of $G'$  is a degree-2 vertex in $G''$ that is adjacent to one vertex of $\cG(x_i)$ and one vertex $c_j$. In other words, each edge incident with $c_j$ corresponds to the first, second, third, or fourth appearance of a variable in $C$. Now, by construction, each vertex of $G''$ has degree at most three with \begin{equation}\label{eq:W}
W = \bigcup_{i=1}^n \{y_{s,i},\, y_{t,i},\, x_{i,1},\, x_{i,2},\, x_{i,3},\, x_{i,4}\}
\end{equation}
being the set of degree-2 vertices of $G''$.

\begin{sublemma}
$C$ is nae-satisfiable if and only if $G''$ has a non-transitive $st$-orientation. 
\end{sublemma}

\begin{proof}
First, suppose that $C$ is nae-satisfiable. Let $\beta\colon X \rightarrow \{T,F\}$ be a truth assignment for $X$ that nae-satisfies $C$. 
Let $V_s$ and $V_t$ be the vertex set of $\cG(s)$ and $\cG(t)$, respectively.
By construction, the induced subgraph $G''[V_s]$ is isomorphic to the source gadget, each vertex in $S = \{y_{s,i} \mid 1 \leq i \leq n\}$ has degree one, and all other vertices have degree three in $G''[V_s]$. 
By Lemma~\ref{lem:source}(i), there is an acyclic and non-transitive orientation $\cD(G''[V_s])$ such that $s$ is the unique vertex with in-degree zero and $S$ is the subset of vertices in $V_s$ 
with out-degree zero. Similarly, the induced subgraph $G''[V_t]$ is isomorphic to the sink gadget, each vertex in $T = \{y_{t,i} \mid 1 \leq i \leq n\}$ has degree one, and all other vertices have degree 
three in $G''[V_t]$. By Lemma~\ref{lem:sink}(i), there is an acyclic and non-transitive orientation $\cD(G''[V_t])$ such that $t$ is the unique vertex with out-degree zero and $T$ is the subset of vertices in $V_t$ with in-degree zero. Turning to the variable gadgets, let $V_{x_i}$ be the vertex set of $\cG(x_i)$. Then the induced subgraph $G''[V_{x_i}]$ is isomorphic to the variable gadget $\cG(x_i)$, each vertex in 
$
\{y_{s,i},\, y_{t,i},\, x_{i,1},\, x_{i,2},\, x_{i,3},\, x_{i,4}\}
$ 
has degree one, and all other vertices have degree three. If $\beta(x_i) = T$, let $\cD(G''[V_{x_i}])$ be the \textit{true} orientation of the variable gadget and, if $\beta(x_i)=F$, let $\cD(G''[V_{x_i}])$ be the \textit{false} orientation \ (see Figure~\ref{fig:variable-gadget-cubic-orientations}). By Lemma~\ref{lem:variable}(i) or (ii), $\cD(G''[V_{x_i}])$ is acyclic and non-transitive, $y_{s,i}$ has in-degree zero and out-degree one, and $y_{t,i}$ has in-degree one and out-degree zero. Moreover, if $\beta(x_i) = T$, each vertex in $\{x_{i,1},\, x_{i,2},\, x_{i,3},\, x_{i,4}\}$ has in-degree one and out-degree zero and, if $\beta(x_i) = F$, these vertices have in-degree zero and out-degree one. Lastly, for each edge $e=\{x_{i,k}, c_{j}\}$ that is incident with a clause gadget $c_j$ in $G''$,
orient $e$ as $(x_{i,k}, c_{j})$ if $\beta(x_i) = T$ and orient $e$ as $(c_{j}, x_{i,k})$ if $\beta(x_i) = F$. Now observe that $\cD(G''[V_s])$, $\cD(G''[V_t])$, $\cD(G''[V_{x_i}])$, and the orientations assigned to the edges incident with a clause gadget collectively assign an orientation to each edge of $G''$. Let $\cD(G'')$ denote this orientation.

We next show that $\cD(G'')$ is a non-transitive $st$-orientation. Clearly, $s$ and $t$ is a  source and sink, respectively, of $\cD(G'')$. Furthermore, by construction and since $\beta$ nae-satisfies each clause $c_j$, each degree-3 vertex except for $s$ and $t$ has in-degree at least one and out-degree at least one. Recall that the set of degree-2 vertices in $G''$ is $W$ (see Equation~\ref{eq:W}) and that each such vertex is incident to two edges that belong to two distinct gadgets.
Now, $y_{s,i}$ has in-degree one due to $\cD(G''[V_s])$ and out-degree one due to $\cD(G''[V_{x_i}])$. Similarly, $y_{t,i}$ has in-degree one and out-degree one due to the orientations chosen for $\cD(G''[V_{x_i}])$ and $\cD(G''[V_t])$. Lastly, the two arcs that are incident with a vertex $x_{i,k}$ form a directed path to a clause gadget vertex $c_j$  if $\beta(x_i) = T$ and a directed path in the reverse direction if $\beta(x_i) = F$. Hence, $s$ and $t$ is the unique source and sink, respectively, of $\cD(G'')$. 
Moreover, since a vertex with in-degree one and out-degree one cannot be incident to a transitive arc, it also follows from the orientations assigned to the edges in $G''$ that $\cD(G'')$ is non-transitive. It remains to show that $\cD(G'')$ is acyclic. Assume towards a contradiction that $\cD(G'')$ contains a directed cycle. By construction, $\cD(G''[V_s])$ does not contain a directed cycle, and the same holds (individually) for each of  $\cD(G''[V_t])$ and $\cD(G''[V_{x_i}])$. Furthermore, the degree-2 vertex $y_{s,i}$  does not lie on a directed cycle because there is no arc from a vertex not in $V_s$ to a vertex in $V_s$. Similarly, the degree-2 vertex $y_{t,i}$  does not lie on a directed cycle because there is no arc from a vertex  in $V_t$ to a vertex not in $V_t$. Hence, there exists an arc $(x_{i,k},c_j)$ or $(c_j,x_{i,k})$ in $\cD(G'')$ that lies on a directed cycle $C''$. Recall that $x_{i,k}$ is a vertex with in-degree one and out-degree one.  If $\beta(x_i) = F$, then it follows from the the \textit{false} orientation $\cD(G''[V_{x_i}])$ that 
$y_{t,i}$ lies on $C''$; a contradiction. Thus, $\beta(x_i)=T$. Then, each directed path starting at $x_{i,k}$ contains $(x_{i,k},c_j)$ and $(c_j,x_{i',k'})$ as its first two arcs for some $i'\in\{1,2,\ldots,n\}$ with  $\beta(x_{i'})=F$ and $k'\in\{1,2,3,4\}$. However, since $\beta(x_{i'})=F$, we again have that $y_{t,i'}$ lies on $C''$; another contradiction. It now follows that $\cD(G'')$ is acyclic and, therefore, a non-transitive $st$-orientation.

Second, suppose that $G''$ has a non-transitive $st$-orientation $\cD(G'')$. We obtain a truth assignment $\beta \colon X \rightarrow \{T,F\}$ by setting $\beta(x_i) = T$ if $(x_{i,1}, c_j)$ is an arc in $\cD(G'')$ and by setting $\beta(x_i) = F$ if $(c_j,x_{i,1})$ is an arc in $\cD(G'')$. Assume towards a contradiction that $C$ has a clause $c_{j} = \{x_p, x_q, x_r\}$ that is not nae-satisfied by $\beta$. We consider the case that $\beta(x_p) = \beta(x_q) = \beta(x_r) = T$ with $p,q,r\in\{1,2,\ldots,n\}$. The other case for which $\beta(x_p) = \beta(x_q) = \beta(x_r) = F$ leads to a contradiction by analogous arguments and is omitted. 
By construction of $\beta$, the digraph $\cD(G'')$ contains the three arcs $(x_{p,1}, c_{j_p})$, $(x_{q,1}, c_{j_q})$ and $(x_{r,1}, c_{j_r})$, where, $j_p$, $j_q$ and $j_r$ is the smallest index of a clause in $C$ that contains $x_p$, $x_q$, and $x_r$, respectively. Consider the arc $(x_{p,1}, c_{j_p})$.  Since $x_{p,1}\notin\{s,t\}$, it follows that $x_{p,1}$ has in-degree one and out-degree one. Thus $(u_1, x_{p,1})$ is an arc in $\cD(G'')$, where $u_1$ is a vertex of the variable gadget $\cG(x_p)$. Let $k = \alpha(x_p,c_j)$. 
By Lemma~\ref{lem:variable}(iii) and since $x_{p,k}\notin\{s,t\}$, it follows that $\cD(G'')$ contains the arcs $(u_k,x_{p,k})$ and $(x_{p,k}, c_j)$. Using the same reasoning, $\cD(G'')$ also contains the arcs $(x_{q,k'}, c_j)$ and $(x_{r,k''}, c_j)$, where $k' = \alpha(x_q,c_j)$ and $k''= \alpha(x_r,c_j)$. Thus $c_j$ is a sink in $\cD(G'')$. As $c_j\ne t$, this gives a contradiction. Hence, $\beta$ nae-satisfies $C$.
\end{proof}

\begin{figure}[t]
    \centering
    \includegraphics[width=\textwidth]{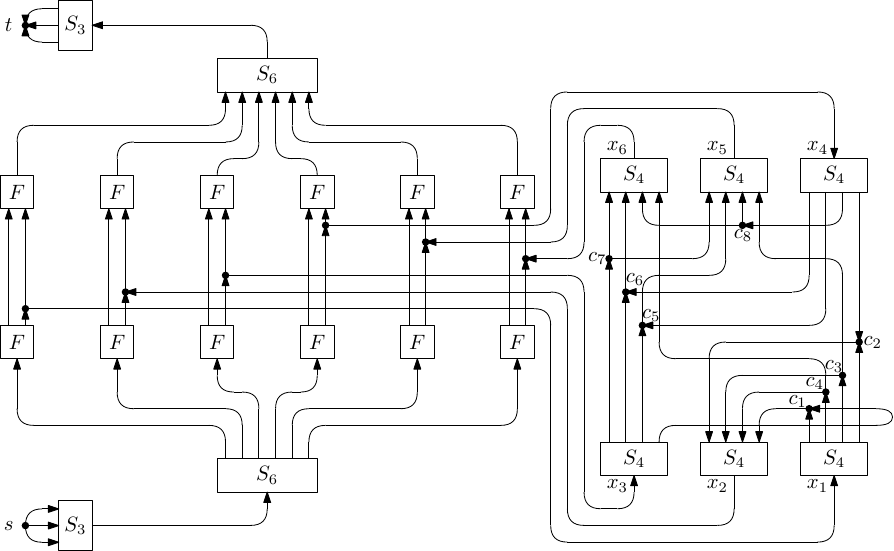}
    \caption{A non-transitive $st$-orientation $\cD(G^*)$ of the cubic graph $G^*$  as in the proof of Theorem~\ref{thm:main-nt-st-orientation} for the example $C_{\text{yes}}$ with clauses $\{x_1, x_2, x_3\}$, $\{x_1, x_2, x_4\}$, $\{x_1, x_2, x_5\}$, $\{x_1, x_2, x_6\}$, $\{x_3, x_4, x_5\}$, $\{x_3, x_4, x_6\}$, $\{x_3, x_5, x_6\}$ and $\{x_4, x_5, x_6\}$. Observe that $\cD(G^*)$ corresponds to the truth assignment $\beta$ with $\beta(x_1) = \beta(x_3) = \beta(x_4) = T$ and $\beta(x_2) = \beta(x_5) = \beta(x_6) = F$.    
    }
    \label{fig:construction-example}
\end{figure} 

Now obtain a graph $G^*$ from $G''$ by suppressing each vertex in $W$. Since $W$ contains every degree-2 vertex of $G''$, the graph $G^*$ is cubic. Furthermore, as $s$ and $t$ are degree-3 vertices in $G$, they are also vertices of $G^*$. Figure~\ref{fig:construction-example} illustrates $G^*$ for the formula $C_{\text{yes}}$ given at the beginning of Section~\ref{sec:nt-st-orientation-cubic}.

\begin{sublemma}
$G''$ has a non-transitive $st$-orientation if and only if $G^*$ has such an orientation.
\end{sublemma}

\begin{proof}
Suppose that $G''$ has a non-transitive $st$-orientation. By construction of $G''$ it is straightforward to check that each vertex in $W$ has two degree-3 neighbors.
Now, let $\cD(G'')$ be a non-transitive $st$-orientation of $G''$, and let $w\in W$. Recalling that $w$ has in-degree one and out-degree one in $\cD(G'')$, let $u$  be the parent of $w$, and let $v$ be the child of $w$ in $\cD(G'')$. We distinguish three cases.  First, if $w =y_{s,i}$, then $u$ is the vertex $s_i$ of $\cG(s)$ and has out-degree one by Lemma~\ref{lem:source}(ii) and (iii). Thus, $w$ is the unique child of $u$. Second, if $w = y_{t,i}$, then $v$ is the vertex $t_i$ of $G(t)$ and has in-degree one by Lemma~\ref{lem:sink}(ii) and (iii). Thus, $w$ is the unique parent of~$v$.
Third, if $w = x_{i,k}$, then $u$ or $v$ is the vertex $u_k$ of $\cG(x_i)$. Moreover, by Lemma~\ref{lem:variable}(iii)--(v), either $w$ is the unique child of $u$ if $u=u_k$, or $w$ is the unique parent of $v$ if $v=u_k$. It now follows that, for all three cases, there is only one directed path connecting $u$ and $v$ in $\cD(G'')$. Hence, by Lemma~\ref{lem:suppress-W}, $G^*$ has a non-transitive $st$-orientation. 
Next suppose that $G^*$ has a non-transitive  $st$-orientation As $G''$ is a subdivision of $G^*$, Lemma~\ref{lem:subdivision} implies that $G''$ has a non-transitive $st$-orientation as well.
\end{proof}

The theorem now follows since the number of vertices in $G''$ is polynomial in $n$ and, so, the construction of $G''$ from $\cI=(X,C)$ and, subsequently, of $G^*$ from $G''$ takes polynomial time. This completes the proof of the theorem.
\end{proof}

Recall the construction of $G''$ in the proof of Theorem~\ref{thm:main-nt-st-orientation}. If we omit the $S_3$ part of the source and the sink gadgets, i.e., we use the split  gadget $S_n$ with $y=s$ as source gadget and the split gadget $S_n$ with $y=t$ as the sink gadget, the following corollary, which will become useful in the next section, follows from Theorem~\ref{thm:main-nt-st-orientation}.

\begin{corollary}\label{cor:st-result-for-phylo}
\NTstO\ is {\em NP}-complete for undirected graphs $G=(V,E)$ with $s,t\in V$ in which $s$ and $t$ have degree one and all other vertices have degree three.
\end{corollary} 
 
\paragraph{\bf \boldmath NP-completeness if $s$ and $t$ are not part of the input.} We conclude this section by remarking that \NTstO\ on cubic graphs remains NP-complete if $s$ and $t$ are not part of the input. More formally, we consider the following decision problem.

\noindent\fbox{%
    \parbox{.97\textwidth}{%
\begin{prb}
\noindent {\sc Relaxed-}\NTstO\\
{\bf Instance.} An undirected graph $G=(V,E)$.\\
{\bf Question.} Does $G$ have a non-transitive $st$-orientation for some choice of vertices $s,t\in V$? \end{prb}
}
}

The source gadget $\cG(s)$ and the sink gadget $\cG(t)$ both contain a cut edge that is not incident with a leaf, i.e., the edge $e$ depicted in Figure~\ref{fig:sosi-gadget-cubic}. These two cut edges, say $e_s$ in $\cG(s)$ and $e_t$ in $\cG(t)$, are also cut edges of the graph $G^*$ that is constructed in the proof of Theorem~\ref{thm:main-nt-st-orientation}. Let $B_s$ (resp.\ $B_t$) denote the biconnected component of $G^*$ that contains $s$ (resp.\ $t$) 
Now, consider a non-transitive $s't'$-orientation $\cD(G^*)$ such that $s'$ and $t'$ are some vertices of $G^*$. Then, one of $B_s$ and $B_t$ contains $s'$ and the other contains $t'$. By Observation~\ref{obs:symmetry}, we may assume that $B_s$ contains $s'$ and $B_t$ contains $t'$. Then, $e_s$ and $e_t$ are oriented in the same way as the cut edge $e$ in the proof of Lemmas~\ref{lem:source} and~\ref{lem:sink}, respectively, which is sufficient for Statements (ii) and (iii) of these two lemmas to still hold. Moreover, as $s'$ and $t'$ are vertices of $\cG(s)$ and $\cG(t)$, respectively, Lemma~\ref{lem:variable} for the variable gadget applies without changes. It now follows that the construction of $G^*$ as presented in the proof of Theorem~\ref{thm:main-nt-st-orientation} also establishes NP-completeness of {\sc Relaxed-}\NTstO\ for cubic graphs. 

\section{Non-transitive orientations of unrooted phylogenetic networks}\label{sec:phylogenetic-networks}

In this section, we show by a reduction from \NTstO\ that deciding the existence of a non-transitive orientation of an unrooted phylogenetic network is NP-complete. We start by defining unrooted (resp.\ rooted) phylogenetic networks and what a (non-transitive) orientation means in this setting.

Let $L$ be a non-empty finite set. An \textit{unrooted binary phylogenetic network} $U$ on $L$ is a leaf-labeled simple connected graph such that each vertex has degree three or degree one and the degree-1 vertices are bijectively labeled with the elements in $L$.  Furthermore, a \textit{rooted binary phylogenetic network} $N$ on $L$ is a leaf-labeled simple acyclic digraph such that there is a unique vertex $\rho$, the \textit{root} of $N$, with in-degree zero and out-degree two,  the vertices with in-degree one and out-degree zero are bijectively labeled with the elements in $L$, and each other vertex has either in-degree one and out-degree two, or in-degree two and out-degree one. If $|L|=1$, then we allow an unrooted as well as a rooted binary phylogenetic network on $L$ to consist of the single vertex in $L$. From now on we omit the term ``binary'' as all unrooted and rooted phylogenetic networks that we consider in this section are binary.

 A rooted phylogenetic network $N$ on $L$ is called an \textit{orientation} of an unrooted phylogenetic network $U$ on $L$ if $|L|=1$ and $U$ and $N$ consist of the single vertex in $L$, or if $N$ can be obtained from $U$ by subdividing an edge of $U$ with the root vertex $\rho$ and then assigning a direction to each edge. As for (unlabeled) directed graphs an arc $(u,v)$ of a rooted phylogenetic network is {\it transitive} if there exists a directed path from $u$ to $v$ that does not traverse $(u,v)$. A rooted phylogenetic network $N$ is \textit{non-transitive} if $N$ does not contain a transitive arc. We remark that in the context of rooted phylogenetic networks, a transitive arc is commonly referred to as a \textit{shortcut}. However, we keep using the term transitive arc throughout this section.  In the remainder of this section, we will investigate the following decision problem.
 
\noindent\fbox{%
    \parbox{.97\textwidth}{%
\begin{prb}
\noindent\PhyloNTO\\
{\bf Instance.} An unrooted phylogenetic network $U$ on $L$.\\
{\bf Question.} Does there exist a rooted phylogenetic network on $L$ that is a non-transitive orientation of $U$? \end{prb}
}
}

To establish NP-completeness of \PhyloNTO, we construct an unrooted phylogenetic network $U$ with exactly one leaf $\ell$ and at least one edge. 
Then a non-transitive orientation of $U$ as a rooted phylogenetic network with root $\rho$ is equivalent to a non-transitive $st$-orientation of $U$ with $s=\rho$ and $t=\ell$.

For the reduction from \NTstO\ to \PhyloNTO, we introduce a gadget that is obtained by a slight modification from the fork gadget. Let the \textit{root gadget} $\cG(\rho)$ be the undirected graph depicted in Figure~\ref{fig:phylo-gadgets}. The notation that is used in the next lemma follows that of the same figure.

\begin{figure}
    \centering
    \includegraphics[width=.8\textwidth]{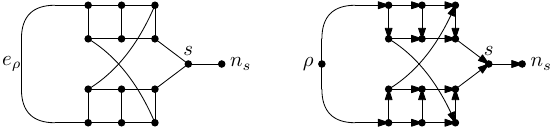}
    \caption{The root gadget (left), and an acyclic and non-transitive orientation of the root gadget (right).}
    \label{fig:phylo-gadgets}
\end{figure}
\begin{lemma}\label{lem:root-gadget}
Let $U$ be an unrooted phylogenetic network on $L$, and let $N$ be a rooted phylogenetic network that is a non-transitive orientation of $U$. Further, let $V$ with $V \cap L = \emptyset$ be a subset of the vertices of $U$ such that the undirected graph $U[V]$ is isomorphic to $\cG(\rho)$. Then $N$ satisfies the following properties:
\begin{enumerate}[\rm (i)]
\item $N$ is obtained from $U$ by subdividing an edge in $U[V]$ and 
\item $(s,n_s)$ is an arc in~$N$.
\end{enumerate}
\end{lemma}

\begin{proof}
Assume towards a contradiction that $N$ is not obtained from $U$ by subdividing an edge that is not in $U[V]$. Since $\{s,n_s\}$ is a cut edge of $U$, $(n_s, s)$ is an arc in $N$. Then, as $N$ is acyclic, the subgraph of $N$ obtained from directing the edges in $U[V]$ contains a vertex with in-degree three. Since this vertex is not a leaf of $N$, we get a contradiction to $N$ being a rooted phylogenetic network.
This establishes (i). Next, assume that $(s,n_s)$ is not an arc in~$N$. Then $(n_s, s)$ is an arc in~$N$ or $\{s,n_s\}$ is subdivided with the root vertex $\rho$, in which case both arcs incident with $\rho$ are directed away from $\rho$. It now follows that we obtain the same contradiction as for (i). Hence, (ii) holds as well.
\end{proof}

\begin{theorem}
\PhyloNTO\ is {\rm NP}-complete.
\end{theorem} 
\begin{proof}
We start by showing that \PhyloNTO\ is in NP. Let $N$ be a directed graph obtained from an unrooted phylogenetic network $U$ on $L$ by subdividing an edge with $\rho$ and then assigning a direction to each edge. By considering, in turn, each vertex of $N$, we can verify in polynomial time that $\rho$ is the unique source and the vertices in $L$ are precisely the sinks of $N$. Checking if $N$ is acyclic and non-transitive can be done in polynomial time using the same arguments as in the proof of Theorem~\ref{thm:main-nt-st-orientation}. Hence, verifying that $N$ is a rooted phylogenetic network on $L$ and a non-transitive orientation of $U$ can be done in polynomial time.

We show NP-hardness of \PhyloNTO\ by a reduction from \NTstO\ for graphs in which the source $s$ and sink $t$ have degree one and all other vertices have degree three which is NP-complete by Corollary~\ref{cor:st-result-for-phylo}. Let $G=(V,E)$ with $s,t\in V$ be such a graph. Furthermore, let $n_s$ denote the unique neighbor of $s$ in $G$, and let $\cG(\rho)$ be the root gadget with vertex set $V_\rho$.
We obtain an undirected graph $U$ from $G$ and $\cG(\rho)$ by first deleting $n_s$ in $\cG(\rho)$ and then identifying $s$ in the resulting graph with $s$ in $G$.
Clearly, no loop or parallel edge is introduced by joining $G$ and $\cG(\rho)$. Thus, $U$ is a simple connected graph. Let $V_U=(V\cup V_\rho)\setminus\{n_s\}$ be the vertex set of $U$. By construction, each vertex in $V_U\setminus \{t\}$ has degree three and $t$ has degree one in $U$. In particular, note that $s$ has two neighbors in $V_\rho\setminus\{n_s\}$ and one neighbor in $V$.
We conclude that $U$ is an unrooted phylogenetic network on $L=\{t\}$.

We complete the proof by showing that  $U$ is a yes-instance of \PhyloNTO\  if and only if $G$ is a yes-instance of \NTstO.  First, suppose that $U$ is a yes-instance of \PhyloNTO. Let $N$ be a rooted phylogenetic network on $L$ that is a non-transitive orientation of $U$. We obtain an $st$-orientation $\cD_\tau(G)$ as follows. For each edge $\{u,v\}$ of $G$, we set $\tau(\{u,v\}) = (u,v)$ if $(u,v)$ is an arc in $N$ and, otherwise, we set $\tau(\{u,v\}) = (v,u)$. Recall that the vertex set of $\cD_\tau(G)$ is equal to the vertex set $V$ of $G$. By Lemma~\ref{lem:root-gadget}(i), $N$ is obtained from $U$ by subdividing an edge of $U[V_\rho\setminus\{n_s\}]$, and so $\cD_\tau(G)$ is well defined. 
It  follows that the in-degree (resp.\ out-degree) of every vertex in $V$ except for $s$ coincides with its in-degree (resp.\ out-degree) in $N$. Hence, $t$ is a sink of $\cD_\tau(G)$ and no vertex in $V\setminus\{s,t\}$ is a source or a sink $\cD_\tau(G)$. Furthermore, by Lemma~\ref{lem:root-gadget}(ii), $(s, n_s)$ is an arc in $N$ and, thus, $s$ is a source of $\cD_\tau(G)$. Hence, $s$ is the unique source and $t$ is the unique sink of $\cD_\tau(G)$. Lastly, as each arc in $\cD_\tau(G)$ is also an arc in $N$, a directed cycle or a transitive arc in $\cD_\tau(G)$ implies that $N$  contains a directed cycle or a transitive arc. Hence, $\cD_\tau(G)$ is a non-transitive $st$-orientation. It now follows that $G$ is a yes-instance of \NTstO.  Second, suppose that $G$ is a yes-instance of \NTstO. Let $\cD_\tau(G)$ be a non-transitive $st$-orientation of $G$. We obtain a rooted phylogenetic network $N$ on $L=\{t\}$ that is an orientation of $U$ by subdividing the edge  $e_\rho$ of the root gadget as illustrated in Figure~\ref{fig:phylo-gadgets} with the vertex $\rho$, directing the edges of the root gadget as shown in the same figure, and directing each other edge $\{u,v\}$ of $U$ as $\tau(\{u,v\})$. As $s$ is the source of $\cD_\tau(G)$, observe that the cut edge $(s,n_s)$ is an arc in $\cD_\tau(G)$. Hence, each edge of $U$ that is also an edge of $G$ obtains the same direction in $N$ as it has in $\cD_\tau(G)$. 
By construction, the root vertex $\rho$ of $N$ has in-degree zero and out-degree two, the unique vertex with in-degree one and out-degree zero is $t$, and each remaining vertex in $N$ has  in-degree at least one and out-degree at least one. 
Furthermore, since $\{s,n_s\}$ is a cut edge of $U$ and neither $\cD_\tau(G)$ nor the orientation of the root gadget as shown in Figure~\ref{fig:phylo-gadgets}  contain a directed cycle, it follows that $N$ is acyclic. Thus, $N$ is a rooted phylogenetic network. Assume towards a contradiction that $N$ has a transitive arc $(u,v)$. The orientation chosen for the root gadget does not contain a transitive arc. In particular, the arc $(s, n_s)$ is not a transitive arc. Hence, both $(u,v)$ and each arc on the directed path from $u$ to $v$ in $N$ that does not traverse $(u,v)$ is also an arc in $\cD_\tau(G)$. Then $\cD_\tau(G)$ contains a transitive arc; a contradiction. We conclude that $N$ is a rooted phylogenetic network that is a non-transitive orientation of $U$ and, thus, $U$ is a yes-instance of \PhyloNTO. 

Since the root gadget has constant size,  the construction of $U$ from $G$ and $\cG(\rho)$ can be done in polynomial time, thereby completing the proof of the theorem. 
\end{proof}

\section{Conclusion}\label{sec:conclusion}
In this paper, we have shown that deciding the existence of a non-transitive $st$-orientation is NP-complete on cubic graphs, thereby settling an open question by Binucci et al.~\cite{binucci23} regarding this problem for graphs with vertex degree at most three. Since  \NTstO\ was previously known to be trivial for graphs whose vertices have degree at most two and NP-complete for graphs whose vertices have degree at most four~\cite{binucci23}, our result closes the remaining gap.
We have also shown that the  problem remains NP-complete  if a source and sink can be chosen freely as opposed to being designated vertices. Building on these results, we have then established NP-completeness for deciding whether an unrooted binary phylogenetic network has a non-transitive orientation as a rooted binary phylogenetic network, a problem that arises in the study of phylogenetic networks.
 
We conclude this paper with four questions for future research:
\begin{enumerate}[(Q$1$)]
\item Is \NTstO\ NP-hard on $k$-regular graphs with $k > 3$?
\item Does \NTstO\ remain NP-hard on cubic graphs with no cycles of length four? Observe that the fork gadget, as introduced in Section~\ref{sec:new-fork-gadget} strongly builds on those cycles.
\item Let $G$ be an undirected graph with vertices $s$ and $t$, and let $q\geq 2$ be a fixed integer. Is it NP-hard to decide whether an $st$-orientation $\cD(G)$ exists such that reversing the direction of any $q$ arcs in $\cD(G)$ does not result in a directed cycle? Note that for $q = 1$ this question is equivalent to \NTstO.
\item Is there a fork gadget that satisfies the topological constraints corresponding to other classes of binary phylogenetic network such as normal networks~\cite{willson10}? More generally, it would be interesting to explore other links between orientations of phylogenetic networks and orientations of  undirected graphs.
\end{enumerate}

\bibliographystyle{alpha}

\begin{thebibliography}{99}

\bibitem{biedl98} 
Biedl, T., Kant, G. (1998). A better heuristic for orthogonal graph drawings. Computational Geometry, 9(3):159--180.

\bibitem{binucci22}
Binucci, C., Didimo, W., Patrignani, M. (2023). $st$-orientations with few transitive edges.  Journal of Graph Algorithms and Applications, 27:625--650.

\bibitem{binucci23}
Binucci, C., Liotta, G., Montecchiani, F., Ortali, G., Piselli, T. (2025). On the parameterized complexity of computing $st$-orientations with few transitive edges. Journal of Graph Algorithms and Applications, 29:647--266.

\bibitem{brightwell93}
Brightwell, G. (1993). On the complexity of diagram testing. Order, 10:297--303.

\bibitem{bulteau23}
Bulteau, L., Weller, M., Zhang, L. (2023). On turning a graph into a phylogenetic network. HAL Open Science hal-04085424.

\bibitem{darmann20}
Darmann, A., Döcker, J. (2020). On a simple hard variant of Not-All-Equal 3-Sat. Theoretical Computer Science, 815:147--152.

\bibitem{dempsey24}
Dempsey, J., van Iersel, L., Jones, M., Murakami, Y., Zeh, N. (2024). A wild sheep chase through an orchard, arXiv:2408.10769.

\bibitem{doecker25}
Döcker, J., Linz, S. (2025). On the existence of funneled orientations for classes of rooted phylogenetic networks. Theoretical Computer Science 1023:114908. 

\bibitem{francis25}
Francis, A. (2025). ``Normal'' phylogenetic networks may be emerging as the leading class. Journal of Theoretical Biology, 614:112236.

\bibitem{garvardt23}
Garvardt, J., Renken, M., Schestag, J., Weller, M. (2023). Finding degree-constrained acyclic orientations. In: 18th International Symposium on Parameterized and Exact Computation, pp. 19:1--19:14.

\bibitem{huber22}
Huber, K. T., van Iersel, L., Janssen, R., Jones, M., Moulton, V., Murakami, Y., Semple, C. (2024). Orienting undirected phylogenetic networks. Journal of Computer and System Sciences 140:103480.

\bibitem{iersel}
van Iersel, L., Jones, M., Linz, S., Zeh, N.. A class of unrooted phylogenetic networks inspired by the properties of rooted tree-child networks, submitted.

\bibitem{janssen18}
Janssen, R., Jones, M. and Erd{\H{o}}s, P. L., van Iersel, L., Scornavacca, C. (2018). Exploring the tiers of rooted phylogenetic network space using tail moves. Bulletin of Mathematical Biology, 80:2177--2208.

\bibitem{kong22}
Kong, S., Pons, J. C., Kubatko, L., Wicke, K. (2022). Classes of explicit phylogenetic networks and their biological and mathematical significance. Journal of Mathematical Biology, 84:47.

\bibitem{maeda23}
Maeda, S., Kaneko, Y., Muramatsu, H., Murakami, Y., Hayamizu, M. (2023). Orienting undirected phylogenetic networks to tree-child networks, arXiv:2305.10162.

\bibitem{nestril95}
Ne\v{s}et\v{r}i, J., R\"odel, V. (1995). More on the complexity of cover graphs. Commentationes Mathematicae Universitatis Carolinae, 36:269--278.

\bibitem{ore62}
Ore, O. (1962) Theory of Graphs, volume 38 of American Mathematical Society Colloquium Publications. American Mathematical Society.

\bibitem{rodel95}
 R\"odel, V., Thoma, L. (1995). The complexity of cover graph recognition for some varieties of finite
lattices. Order, 12:351--374.

\bibitem{schaefer78}
Schaefer, T. J. (1978). The complexity of satisfiability problems. In: Proceedings of the Tenth annual ACM Symposium on Theory of Computing, pp. 216--226.

\bibitem{urata24}
Urata, T., Yokoyama, M., Hayamizu, M. (2024). Orientability of undirected phylogenetic networks to a desired class: Practical algorithms and application to tree-child orientation. In: 24th International Workshop on Algorithms in Bioinformatics, WABI 2024, LIPIcs 312, pp. 9:1--9:17.

\bibitem{weiskircher01} Weiskircher, R. (2001). Drawing planar graphs. In: Drawing Graphs: Methods and Models, pp. 23--45. 

\bibitem{willson10}
Willson, S. J. (2010). Properties of normal phylogenetic networks. Bulletin of
Mathematical Biology, 72:340--358.

\end{thebibliography}

\end{document}